\documentclass[11pt]{article}
\usepackage[latin1]{inputenc}
\usepackage{times,mathptmx}
\usepackage{latexsym}
\usepackage{comment}
\usepackage{amsmath}
\usepackage{amsthm}
\usepackage{amssymb}
\usepackage{graphicx}
\usepackage{ifthen}
\usepackage{calc}
\RequirePackage{fullpage}

\usepackage[margin=1in]{geometry}  

\usepackage{stmaryrd}
\usepackage{qtree}

\numberwithin{equation}{section}

\usepackage{pstricks}
\usepackage{pst-node}
\usepackage{pst-tree}
\usepackage{pst-text}
\usepackage{rotating}

\psset{dotsize=2mm,linewidth=0.7pt}
\newpsobject{showgrid}{psgrid}{subgriddiv=1,griddots=10,gridlabels=6pt}

\newtheoremstyle{mythm}
  {}
  {}
  {\itshape}
  {}
  {\bfseries}
  {.}
  {.5em}
  {\thmname{#1}~\thmnumber{#2}\ifthenelse{\equal{\thmnote{#3}}{}}{}{~(\thmnote{#3})}}

\newtheoremstyle{mydefn}
  {}
  {}
  {\upshape}
  {}
  {\bfseries}
  {.}
  {.5em}
  {\thmname{#1}~\thmnumber{#2}\ifthenelse{\equal{\thmnote{#3}}{}}{}{~(\thmnote{#3})}}

\theoremstyle{mythm}
\newtheorem{theo}{Theorem}
\newtheorem{lem}[theo]{Lemma}
\newtheorem{prop}[theo]{Proposition}
\newtheorem{cor}[theo]{Corollary}
\newtheorem{fact}[theo]{Fact}

\theoremstyle{mydefn}

\theoremstyle{remark}

\newcommand{\ceil}[1]{\left\lceil#1\right\rceil}
\newcommand{\floor}[1]{\left\lfloor#1\right\rfloor}
\renewcommand{\phi}{\varphi}

\newcommand{\dom}{\operatorname{dom}}

\newcommand{\tup}{\operatorname{tup}}

\newcommand{\Exp}{\operatorname{E}}
\newcommand{\Var}{\operatorname{V}}
\newcommand{\Prob}{\operatorname{Pr}}
\newcommand{\density}{\operatorname{\delta}}
\newcommand{\maxdensity}{\operatorname{\overline{\delta}}}

\newenvironment{example}{\bigskip\noindent\textit{Example}:\ }{\bigskip}

\begin{document}
\title{
{\bf Size bounds and query plans for relational joins}\footnote{A
  preliminary version of this paper appeared under the same title in
  the Proceedings of 49th IEEE Symposium on Foundations of Computer
  Science (FOCS), pp. 739-748, 2008.}}

\author{ Albert Atserias\footnote{First author partially supported by
    CYCIT TIN2010-20967-C04-05 (TASSAT).} \\ Universitat Polit\`ecnica
  de Catalunya \\ Barcelona, Spain \and Martin Grohe \\ Humboldt
  Universit\"at zu Berlin \\ Berlin, Germany \and D\'aniel Marx\footnote{Third author supported by the European Research Council (ERC)  grant 280152.}
  \\ Computer and Automation Research Institute,\\ Hungarian Academy of Sciences
(MTA SZTAKI),\\ Budapest, Hungary. } \maketitle

\begin{abstract}
  Relational joins are at the core of relational algebra, which in turn is the
  core of the standard database query language SQL. As their evaluation is
  expensive and very often dominated by the output size, it is an important
  task for database query optimisers to compute estimates on the size of joins
  and to find good execution plans for sequences of joins. We study these
  problems from a theoretical perspective, both in the worst-case model, and
  in an average-case model where the database is chosen according to a known
  probability distribution. In the former case, our first key observation is
  that the worst-case size of a query is characterised by the fractional edge
  cover number of its underlying hypergraph, a combinatorial parameter
  previously known to provide an upper bound.  We complete the picture by
  proving a matching lower bound, and by showing that there exist queries for
  which the join-project plan suggested by the fractional edge cover approach
  may be substantially better than any join plan that does not use
  intermediate projections.
  On the other hand, we show that in the
  average-case model, every join-project plan can be turned into a
  plan containing no projections in such a way that the expected time
  to evaluate the plan increases only by a constant factor
  independent of the
  size of the database. Not surprisingly, the key combinatorial
  parameter in this context is the maximum density of the underlying 
  hypergraph. We show how to make effective use of this parameter to
  eliminate the projections.
\end{abstract}





\section{Introduction}
\label{sec:intro}

The join operation is one of the core operations of
relational algebra, which in turn is the core of the
standard database query language SQL. The two key components
of a database system executing SQL-queries are the query
optimiser and the execution engine. The optimiser translates
the query into several possible execution plans, which are
basically terms of the relational algebra (also called
operator trees) arranging the operations that have to be
carried out in a tree-like order. Using statistical
information about the data, the optimiser estimates the
execution cost of the different plans and passes the best one
on to the execution engine, which then executes the plan and
computes the result of the query. See~\cite{chau98} for a
survey of query optimisation techniques.

Among the relational algebra operations, joins are usually
the most costly, simply because a join of two relations,
just like a Cartesian product of two sets, may be much
larger than the relations. Therefore, query optimisers pay
particular attention to the execution of joins, especially
to the execution order of sequences of joins, and to
estimating the size of joins. In this paper, we address the
very fundamental questions of how to estimate the size of a
sequence of joins and how to execute the sequence best from
a theoretical point of view. While these questions have been
intensely studied in practice, and numerous heuristics and
efficiently solvable special cases are known (see, e.g.,
\cite{chau98,gra93,garullwid99}), the very basic theoretical
results we present here and their consequences apparently
have not been noticed so far. Our key starting observation is that
the size of a sequence 
of joins is tightly linked to two combinatorial 
parameters of the underlying database schema,
the \emph{fractional edge cover number}, and 
the \emph{maximum density}.

To make this precise, we need to get a bit more technical: A
\emph{join query} $Q$ is an expression of the form
\begin{equation}
R_1(a_{11},\ldots,a_{1r_1})\bowtie\cdots\bowtie R_m(a_{m1},\ldots,a_{mr_m}),
\label{eqn:joinquery}
\end{equation}
where the $R_i$ are \emph{relation names} with \emph{attributes}
$a_{i1},\ldots,a_{ir_i}$. Let $A$ be the set of all attributes
occurring in $Q$ and $n = |A|$.  A \emph{database instance} $D$ for
$Q$ consists of relations $R_i(D)$ of arity $r_i$. It is common to
think of the relation $R_i(D)$ as a table whose columns are labelled
by the attributes $a_{i1},\ldots,a_{ir_i}$ and whose rows are the
tuples in the relation. The \emph{answer}, or \emph{set of
  solutions}, of the query $Q$ in $D$ is the $n$-ary relation $Q(D)$
with attributes $A$ consisting of all tuples $t$ whose projection on
the attributes of $R_i$ belongs to the relation $R_i(D)$, for all
$i$. Hence we are considering \emph{natural joins} here (all of our
results can easily be transferred to \emph{equi-joins}, but not to
general \emph{$\theta$-joins}).  Now the most basic question is how
large $Q(D)$ can get in terms of the size of the database $|D|$, or
more generally, in terms of the sizes of the relations $R_i$. We
address this question both in the worst case and the average case, and
also subject to various constraints imposed on $D$.

\begin{example}
At this point a simple example would probably help to understand what
we are after. Let $R(a,b)$, $S(b,c)$ and $T(c,a)$ be three relations
on the attributes $a$, $b$ and $c$.  Consider the join query
\[
Q(a,b,c) := R(a,b) \Join S(b,c) \Join T(c,a).
\] 
The answer of $Q$ is precisely the set of triples $(u,v,w)$ such that
$(u,v) \in R$, $(v,w) \in S$ and $(w,u) \in T$. How large can the
answer size of $Q$ get as a function of $|R|$, $|S|$ and $|T|$? First
note that a trivial upper bound is $|R| \cdot |S| \cdot |T|$.  However
one quickly notices that an improved bound can be derived from the
fact that the relations in $Q$ have overlapping sets of
attributes. Indeed, since any solution for any pair of relations in
$Q$ determines the solution for the third, the answer size of $Q$ is
bounded by $\min\{ |R| \cdot |S|, |S| \cdot |T|, |T| \cdot |R|
\}$. Now, is this the best general upper bound we can get as a
function of $|R|$, $|S|$ and $|T|$? As it turns out, it is
not. Although not obvious, it will follow from the results in this
paper that the optimal upper bound in this case is $\sqrt{|R| \cdot
  |S| \cdot |T|}$: the answer size of $Q$ is always bounded by this
quantity, and for certain choices of the relations $R$, $S$, $T$, this
upper bound is achieved.
\end{example}

Besides estimating the answer size of join queries, we also study how
to exploit this information to actually compute the query.  An
\emph{execution plan} for a join query describes how to carry out the
evaluation of the query by simple operations of the relational algebra
such as joins of two relations or projections. The obvious execution
plans for a join query break up the sequence of joins into pairwise
joins and arrange these in a tree-like fashion. We call such execution
plans \emph{join plans}. As described in \cite{chau98}, most practical
query engines simply arrange the joins in some linear (and not even a
tree-like) order and then evaluate them in this order.  However, it is
also possible to use other operations, in particular projections, in
an execution plan for a join query. We call execution plans that use
joins and projections \emph{join-project plans}. It is one of our main
results that, even though projections are not necessary to evaluate
join queries, their use may speed up the evaluation of a query
super-polynomially.

\subsection*{Fractional covers, worst-case size, and join-project plans}
Recall that an \emph{edge cover} of a hypergraph $H$ is a set $C$ of
edges of $H$ such that each vertex is contained in at least one edge
in $C$, and the \emph{edge cover number} $\rho(H)$ of $H$ is the
minimum size among all edge covers of $H$. A \emph{fractional edge
  cover} of $H$ is a feasible solution for the linear programming
relaxation of the natural integer linear program describing edge
covers, and the \emph{fractional edge cover number} $\rho^*(H)$ of $H$
is the cost of an optimal solution. With a join query $Q$ of the form
(\ref{eqn:joinquery}) we can associate a hypergraph $H(Q)$ whose
vertex set is the set of all attributes of $Q$ and whose edges are the
attribute sets of the relations $R_i$. The \emph{(fractional) edge
  cover number} of $Q$ is defined by $\rho(Q)=\rho(H(Q))$ and
$\rho^*(Q)=\rho^*(H(Q))$.  Note that in the example of the previous
paragraph, the hypergraph $H(Q)$ is a triangle. Therefore in that case
$\rho(Q) = 2$ while it can be seen that $\rho^*(Q) = 3/2$.

An often observed fact about edge covers is that, for every given
database $D$, the size of $Q(D)$ is bounded by $|D|^{\rho(Q)}$, where
$|D|$ is the total number of tuples in $D$.  Much less obvious is the
fact that the size of $Q(D)$ can actually be bounded by
$|D|^{\rho^*(Q)}$, as proved by the second and third author
\cite{gromar06} in the context (and the language) of constraint
satisfaction problems. This is a consequence to Shearer's
Lemma~\cite{MR859293}, which is a combinatorial consequence of the
submodularity of the entropy function, and is closely related to a
result due to Friedgut and Kahn~\cite{frikah98} on the number of
copies of a hypergraph in another.  Our first and most basic
observation is that the fractional edge cover number $\rho^*(Q)$ also
provides a lower bound to the worst-case answer size: we show that for
every $Q$, there exist arbitrarily large databases $D$ for which the
size of $Q(D)$ is at least $(|D|/|Q|)^{\rho^*(Q)}$.  The proof is a
simple application of linear programming duality.
%
%
Another result from \cite{gromar06} implies that for every join query
there is a join-project plan, which can easily be obtained from the query and
certainly be computed in polynomial time, that computes $Q(D)$ in time
$O(|Q|^2\cdot|D|^{\rho^*(Q)+1})$. Our lower bound shows that this is optimal up
to a polynomial factor (of $|Q|^{2+\rho^*(Q)}\cdot|D|$, to be precise). In
particular, we get the following equivalences giving an exact combinatorial
characterisation of all classes of join queries that have polynomial size
answers and can be evaluated in polynomial time.

\begin{theo} \label{theo:intro}
  Let $\mathcal Q$ be a class of join queries. Then the following statements
  are equivalent:
  \vspace{-0.0ex}
  \begin{enumerate} \itemsep=0pt\parsep=0pt
  \item Queries in $\mathcal Q$ have answers of polynomial size.
  \item Queries in $\mathcal Q$ can be evaluated in polynomial time.
  \item Queries in $\mathcal Q$ can be evaluated in
  polynomial time by an explicit join-project plan.
  \item $\mathcal Q$ has bounded fractional edge cover number.
  \end{enumerate}
\end{theo}
Note that it is not even obvious that the first two
statements are equivalent, that is, that for every class of
queries with polynomial size answers there is a polynomial
time evaluation algorithm (the converse, of course, is
trivial).

Hence with regard to worst-case complexity, join-project plans are
optimal (up to a polynomial factor) for the evaluation of join
queries. Our next result is that join plans are not: We prove that
there are arbitrarily large join queries $Q$ and database instances
$D$ such that our generic join-project plan computes $Q(D)$ in at most
cubic time, whereas any join plan requires time $|D|^{\Omega(\log
  |Q|)}$ to compute $Q(D)$. We also observe that this bound is tight,
i.e., the ratio of the exponents between the best join plan and the
best join-project plan is at most logarithmic in $|Q|$. Hence
incorporating projections into a query plan may lead to a
superpolynomial speed-up even if the projections are completely
irrelevant for the query answer.

\subsection*{Maximum density, average-case size, and join plans}
Consider the model $\mathcal D(N,(p_R))$ of random databases where the tuples
in each relation $R$ are chosen randomly and independently with probability
$p_R = p_R(N)$ from a domain of size $N$. This is the analogue of the
Erd\H os-R\'enyi model of random graphs adapted to our context. It is easy to
see that, for $D$ from $\mathcal D(N,(p_R))$, the expected size of the query
answer $Q(D)$ is $N^n\cdot\prod_R p_R$, where $n$ is the number of attributes
and the product ranges over all relation names $R$ in $Q$. The question is
whether $|Q(D)|$ will be concentrated around the expected value.
 This is governed by the \emph{maximum density}
$\maxdensity(Q,(p_R))$ of the query, a combinatorial parameter depending on
the hypergraph of the query and the probabilities $p_R$.  An application of the
second moment method shows that if $\maxdensity=\log N-\omega(1)$, then
$|Q(D)|$ is concentrated around its expected value, and if $\maxdensity=\log
N+\omega(1)$, then $|Q(D)|=0$ almost surely.  We observe that the maximum
density $\maxdensity$ can be computed in polynomial time using max-flow min-cut
techniques. 

In view of the results about the worst-case, it is a
natural question whether join-project plans are more powerful than join plans
in the average case setting as well.  It turns out that this is not the case:
We show that every join-project plan $\phi$ for $Q$ can be turned into a join
plan $\phi'$ for which the expected execution time increases only by a
constant factor independent of the database. This may be viewed as our main
technical result.  The transformation of
$\phi'$ into $\phi$ depends on a careful balance between delaying certain
joins in order to reduce the number of attributes considered in each subquery
occurring in the plan and keeping as many joins as possible in order to
increase the density of the subquery. The choice of which subqueries to delay
and which to keep is governed by a certain submodular function related to the
density of the subqueries.


\subsection*{Size and integrity constraints}
So far, we considered worst-case bounds which make no assumptions on
the database, and average-case bounds which assume a known
distribution on the database. However, practical query optimisers
usually exploit additional information about the databases when
computing their size estimates. We consider the simplest such setting
where the sizes of the relations are known (called histograms in the
database literature), and we want to get a (worst case) estimate on
the size of $Q(D)$ subject to the constraint that the relations in $D$
have the given sizes.

By suitably modifying the
objective function of the linear program for edge covers, we obtain
results analogous to those obtained for the unconstrained setting. A
notable difference between the two results is that here the gap
between upper and lower bound becomes $2^{-n}$, where $n$ is the
number of attributes, instead of $|Q|^{-\rho^*}$. We give an example
showing that the gap between upper and lower bound is essentially
tight. However, this is not an inadequacy of our approach through
fractional edge covers, but due to the inherent complexity of the
problem: by a reduction from the maximum independent-set problem on
graphs, we show that, unless $\textup{NP}=\textup{ZPP}$, there is no
polynomial time algorithm that approximates the worst case answer size
$|Q(D)|$ for given $Q$ and relation sizes $N_R$ by a
better-than-exponential factor.  


Besides the actual sizes of the relations, one could consider other
pieces of information that are relevant for estimating the answer size
of a query, such as functional dependencies or other integrity
constraints that the databases may be specified to satisfy. For
example, if an attribute or a set of attributes plays the role of a
key in a relation, then the size of that relation is bounded by the
size of its projection on the key-attributes, and therefore it
suffices to analyse the contribution of those attributes to the
maximum answer size of the query. In the preliminary version of this
paper we announced some partial results in this direction for the case
of simple functional dependencies. Since then, the problem of
analysing the answer size subject to general functional dependencies
has been addressed in its own right in the more recent works
\cite{GottlobLeeValiant2009} and \cite{ValiantValiant2010}.

\subsection*{Organization}
In Section~\ref{sec:preliminaries} we introduce notation and the basic
definitions. In Section~\ref{sec:worst} we state and prove the bounds
in the worst-case model. Lemmas~\ref{lem:gm} and~\ref{lem:lower} state
the upper bound and the lower bound,
respectively. Theorem~\ref{theo:alg} states the fact that, for queries
of bounded fractional edge cover number, join-project plans can
evaluate the query in polynomial time, and Theorem~\ref{theo:jp}
states that, in contrast, join-only plans cannot. In
Section~\ref{sec:constraints} we incorporate size-constraints into the
analysis. Theorem~\ref{theo:size-bounds} states the upper and lower
bounds for this case, and Theorem~\ref{theo:approxhardness} states
that approximating the maximum output-size better than what
Theorem~\ref{theo:size-bounds} gives is NP-hard. In
Section~\ref{sec:average} we study the average-case model. In
Theorems~\ref{thm:avgcase} and~\ref{thm:avgcase2} we estimate the
output-size as a function of the maximum density of the query. In
Theorem~\ref{thm:removeproject} we show how to exploit the
average-case model to remove projections from any join-project plan
without affecting the run-time by more than a constant factor, on
average.

\section{Preliminaries}
\label{sec:preliminaries}

For integers $m\le n$, by
$[m,n]$ we denote the set $\{m,m+1,\ldots,n\}$ and by $[n]$ we denote $[1,n]$. 
All our logarithms are base $2$.

Our terminology is similar to that used in
\cite{abihulvia95}: An \emph{attribute} is a symbol $a$ with an associated
\emph{domain} $\dom(a)$. If not specified otherwise, we assume $\dom(a)$ to be
an arbitrary countably infinite set, say, $\mathbb N$. Sometimes, we will
impose restrictions on the size of the domains.
A \emph{relation name} is a symbol $R$ with an associated finite set
of attributes $A$. For a set $A=\{a_1,\ldots,a_n\}$ of attributes, we
write $R(A)$ or $R(a_1,\ldots,a_n)$ to denote that $A$ is the set of
attributes of $R$. The \emph{arity} of $R(A)$ is $|A|$. A
\emph{schema} is a finite set of relation names. If $\sigma =
\{R(A_1),\ldots,R(A_m)\}$, we write $A_\sigma$ for $\bigcup_i A_i$.

For a set $A$ of attributes, an \emph{$A$-tuple} is a mapping $t$ that
associates an element $t(a)$ from $\dom(a)$ with each $a\in A$. Occasionally,
we denote $A$-tuples in the form $t=(t_a : a\in A)$, with the obvious meaning
that $t$ is the $A$-tuple with $t(a)=t_a$. The set of all $A$-tuples is
denoted by $\tup(A)$. An \emph{$A$-relation} is a set of $A$-tuples. 
The \emph{active domain} of an $A$-relation $R$ is the set $\{t(a) :
t\in R, a\in A\}$.  The \emph{projection} of an $A$-tuple $t$ to a subset
$B\subseteq A$ is the restriction $\pi_B(t)$ of $t$ to $B$, and the
\emph{projection} of an $A$-relation $R$ is the set $\pi_B(R)=\{\pi_B(t) :
t\in R\}$.

A \emph{database instance} $D$ of schema $\sigma$, or a
\emph{$\sigma$-instance}, consists of an $A$-relation $R(D)$ for
every relation name $R$ in $\sigma$ with set of attributes $A$. 
The \emph{active domain} of $D$
is the union of active domains of all its relations.
The \emph{size} of a $\sigma$-instance $D$ is 
$|D|:=\sum_{R\in\sigma}|R(D)|$.

A \emph{join query} is an expression
\[
Q:=R_1(A_1)\bowtie\cdots\bowtie R_m(A_m),
\]
where $R_i$ is a relation name with attributes $A_i$.  The \emph{schema} of
$Q$ is the set $\{R_1,\ldots,R_m\}$, and the \emph{set of attributes} of $Q$
is $\bigcup_{i} A_i$. We often denote the set of attributes of a join query
$Q$ by $A_Q$, and we write $\tup(Q)$ instead of
$\tup(A_Q)$. 
The
\emph{size} of $Q$ is $|Q|:=\sum_{i} |A_i|$. We write $H(Q)$ for the
(multi-)hypergraph that has vertex-set $A_Q$ and edge-(multi-)set
$\{A_1,\ldots,A_m\}$.
If $D$ is an $\{R_1,\ldots,R_m\}$-instance, the \emph{answer} of $Q$ on $D$ 
is the $A_Q$-relation
\[
Q(D)=\big\{ t\in \tup(A_Q) : \pi_{A_i}(t)\in R_i(D)\text{
  for every } i \in [m] \big\}.
\]
A \emph{join plan} is a
term built from relation names and binary join operators. For example,
$(R_1\bowtie R_2)\bowtie(R_3\bowtie R_4)$ and $((R_1\bowtie R_2)\bowtie
R_3)\bowtie (R_1\bowtie R_4)$ are two join plans corresponding to the same
join query $R_1\bowtie R_2\bowtie R_3\bowtie R_4$. A \emph{join-project plan}
is a term built from relation names, binary join operators, and unary project
operators. For example, $(\pi_A(R_1)\bowtie R_2)\bowtie \pi_B(R_1)$ is a
join-project plan.  Join-project plans have a natural representation as
labelled binary trees, where the leaves
are labelled by relation names, the
unary nodes are labelled by projections $\pi_A$, and the binary nodes by
joins.  Evaluating a join plan or join-project plan $\phi$ in a database
instance $D$ means substituting the relation names by the actual relations
from $D$ and carrying out the operations in the expression. We denote the
resulting relation by $\phi(D)$.  A join(-project) plan $\phi$ is a plan
\emph{for} a query $Q$ if $\phi(D)=Q(D)$ for every database $D$. The
\emph{subplans} of a join(-project) plan are defined in the obvious way. For
example, the subplans of $(R_1\bowtie R_2)\bowtie\pi_A(R_3\bowtie R_4)$ are
$R_1$, $R_2$, $R_3$, $R_4$, $R_1\bowtie R_2$, $R_3\bowtie R_4$,
$\pi_A(R_3\bowtie R_4)$, $(R_1\bowtie R_2)\bowtie\pi_A(R_3\bowtie R_4)$. If
$\phi$ is a join project plan, then we often use $A_\phi$ to denote the set of
attributes of the query computed by $\phi$ (this only includes ``free''
attributes and not those projected away by some projection in $\phi$), and we
write $\tup(\phi)$ instead of $\tup(A_\phi)$.

\section{Worst-case model}
\label{sec:worst}

In this section, we study the worst-case model in which we make no
assumptions at all on the database. First we discuss the estimates on
the answer-size of join queries, and then we address the question of
query plans for such queries.

\subsection{Size bounds}\label{sec:wc-bounds}
Let $Q$ be a join query with schema $\sigma$. For every $R \in
\sigma$, let $A_R$ be the set of attributes of $R$, so that $A_\sigma
= \bigcup_R A_R$.  The fractional edge covers are precisely the
feasible solutions $(x_R : R \in \sigma)$ for the following linear
program $L_Q$, and the fractional edge cover number $\rho^*(Q)$ is the
cost of an optimal solution.
\begin{equation}
\begin{array}{llll}
  L_Q:\quad&\text{minimise } & \sum_R x_R & \\
  &\text{subject to}
  & \sum_{R\;:\;a \in A_R} x_R \ge 1
  & \text{for all } a \in A_\sigma,\\
  && x_R \ge 0 & \text{for all } R \in \sigma.
\end{array}
\label{eqn:lp}
\end{equation}
By standard arguments, there always is an optimal fractional edge
cover whose values are rational and of bit-length polynomial in $|Q|$.
As observed in \cite{gromar06}, fractional edge covers can be used to
give an upper bound on the size of a query.  
\begin{lem}[\cite{gromar06}]\label{lem:gm}
  Let $Q$ be a join query with schema $\sigma$ and 
  let $D$ be a $\sigma$-instance. Then
  for every fractional edge cover $(x_R : R \in \sigma)$ of $Q$ 
  we have
  \begin{equation*}
    |Q(D)| \le \prod_{R\in\sigma}|R(D)|^{x_R}=2^{\sum_{R\in\sigma}x_R\log|R_D|}.
  \end{equation*}
\end{lem}
Note that the fractional edge cover in the statement of the lemma is not
necessarily one of minimum cost. For the reader's
convenience, we give a proof of this lemma, which is actually a
simplification of the proof in~\cite{gromar06}.

The proof of Lemma~\ref{lem:gm} is based on a combinatorial lemma known as
Shearer's lemma. The lemma appeared first in \cite{MR859293}, where it is attributed to
Shearer. The \emph{entropy} of
a random variable
$X$ with range $U$ is 
\[
h[X]:=- \sum_{x \in U}
  \Prob[ X=x ]\log \Prob[ X = x ]
\]
Shearer's lemma gives an upper bound on the entropy of a distribution on a product space in
terms of the entropies of its marginal distributions. 

\begin{lem}[Shearer's Lemma]
  Let $X=(X_i\mid i\in I)$ be a random variable, and let $A_j$, for $j\in J$, be
  (not necessarily distinct) subsets of the index set $I$ such that each $i\in
  I$ appears in at least $k$ of the sets $A_j$. For every $B\subseteq I$, let
  $X_B=(X_i\mid i\in B)$. Then
  \[
  \sum_{j=1}^m h[X_{A_j}]\ge k\cdot h[X].
  \]
\end{lem}

A simple proof of the lemma can be found in \cite{rha01}.

Now we are ready to prove Lemma~\ref{lem:gm}:

\begin{proof}[Proof of Lemma~\ref{lem:gm}]
  Let $A_R$ be the set of attributes of $R \in \sigma$ so that
  $A_\sigma = \bigcup_R A_R$. Without loss of generality we may assume
  that the fractional edge cover $x_R$ only takes rational values,
  because the rationals are dense in the reals.  Let $p_R$ and $q$ be
  nonnegative integers such that $x_R = p_R/q$. Let $m = \sum_R p_R$,
  and let $A_1,\ldots,A_m$ be a sequence of subsets of $A_\sigma$ that
  contains precisely $p_R$ copies of the set $A_R$, for all
  $R\in\sigma$. Then every attribute $a \in A_\sigma$ is contained in
  at least $q$ of the sets $A_i$, because
  \[
  |\{i\in[m] : a\in A_i\}\big| = \sum_{R : a \in A_R} p_R
  \;=\; q \cdot \sum_{R : a \in A_R} x_R \;\geq\; q.
  \]
  Let $X=(X_a\mid a\in A_\sigma)$ be uniformly distributed on $Q(D)$,
  which we assume to be non-empty as otherwise the claim is
  obvious. That is, for every tuple $t\in Q(D)$ we have $\Prob[ X=t ]
  =1/|Q(D)|$, and for all other $A$-tuples we have $\Prob[ X=t ] =0$.
  Then $h[X]=\log|Q(D)|$. We apply Shearer's Lemma to the random
  variable $X$ and the sets $A_R$, for $R\in\sigma$. (Thus we have
  $I=A_\sigma$ and $J=\sigma$.)  Note that for every $R\in\sigma$ the
  marginal distribution of $X$ on $A_R$ is $0$ on all tuples not in
  $R({D})$. Hence the entropy of $X_{A_R}$ is bounded by the entropy
  of the uniform distribution on $R(D)$, that is,
  $h[X_{A_R}]\le\log|R(D)|$.  Thus by Shearer's Lemma, we have
  \[
  \sum_{R\in\sigma} p_R\cdot\log|R(D)|\ge
\sum_{R\in\sigma} p_Rh[X_{A_R}]=
\sum_{i=1}^m h[X_{A_i}]\ge q\cdot
  h[X]=q\cdot\log|Q(D)|.
  \]
  It follows that 
  \[
  |Q(D)|\le 2^{\sum_{R\in\sigma} (p_R/q)\cdot\log|R({
      D})|}=\prod_{R\in \sigma} {|R(D)|}^{x_R}.
  \]
\end{proof}

The next lemma shows that the upper bound of the previous lemma is tight:

\begin{lem} \label{lem:lower}
  Let $Q$ be a join query with schema  $\sigma$, and let
  $(x_R : R \in \sigma)$ be an optimal fractional edge cover of $Q$.  
  Then for every
  $N_0\in\mathbb N$ there is a $\sigma$-instance $D$ such that $|D|\ge
  N_0$ and
  \[
  |Q(D)|\ge \prod_{R \in \sigma} |R(D)|^{x_R}.
  \]
  Furthermore, we can choose $D$ in such a way that $|R(D)|=|R'(D)|$ for all
  $R,R'\in\sigma$ with $x_R,x_{R'}>0$.
\end{lem}
%
\begin{proof}
  Let $A_R$ be the set of attributes of $R \in \sigma$ so that
  $A_\sigma = \bigcup_R A_R$.  Recall $(x_R : R \in \sigma)$ is an
  optimal solution for the linear program (\ref{eqn:lp}).  By
  LP-duality, there is a solution $(y_a : a \in A_\sigma)$ for the
  dual linear program
  \begin{equation}  \label{eqn:dual}
  \begin{array}{lll}
    \text{maximise} & \sum_{a} y_a & \\
    \text{subject to} &
    \sum_{a\in A_R} y_a \le 1 & \text{for all } R \in
    \sigma, \\
    & y_a \ge 0 & \text{for all } a\in A_\sigma
  \end{array}
  \end{equation}
  such that $\sum_{a} y_a=\sum_{R} x_R$. There even
  exists such a solution with rational values.
  
  We take an optimal solution $(y_a : a\in A_\sigma)$ with
  $y_a=p_a/q$, where $q \ge 1$ and $p_a \ge 0$ are integers. Let $N_0
  \in \mathbb N$, and let $N = N_0^q$. We define a $\sigma$-instance
  $D$ by letting
  \[
  R(D):=\big\{t\in \tup(A_R) : t(a) \in [N^{p_a/q}] \text{ for
  all } a \in A_R \big\}
  \]
  for all $R\in\sigma$. 
  Here we assume that $\dom(a)=\mathbb N$ for all attributes
  $a$. As there is at least one $a$ with $y_a>0$ and hence
  $p_a\ge 1$, we have $|D|\ge N^{1/q}=N_0$. Observe that
  \begin{equation*}
  |R(D)|=\prod_{a\in A_R}N^{p_a/q}=N^{\sum_{a \in A_R} y_a}\le N
  \label{eqn:size}
  \end{equation*}
  for all $R\in\sigma$. Furthermore, $Q(D)$ is the set of all tuples
  $t \in \tup(A_\sigma)$ with $t(a) \in [N^{p_a/q}]$ for every $a \in
  A_\sigma$.  Hence
  \[
  |Q(D)|=\prod_{a\in{A}}N^{p_a/q}=N^{\sum_{a\in A_\sigma} y_a}=
  N^{\sum_{R\in\sigma} x_R} = \prod_{R \in \sigma} N^{x_R}\ge \prod_{R
    \in \sigma} |R(D)|^{x_R},
  \]
as required. To see that $|R(D)|$ is the same for every relation $R$ with $x_R>0$, we argue as follows.
  By complementary slackness of linear programming we have
  \begin{equation*}
    \label{eq:slack}
    \sum_{a\in A_R} y_a=1\quad\text{for all $R\in\sigma$ with $x_R>0$}.
  \end{equation*}
  Thus $|R(D)|=N$ for all $R\in\sigma$ with $x_R>0$ and
  \[
  |Q(D)| = \prod_{R \in \sigma} N^{x_R}= \prod_{R \in \sigma} |R(D)|^{x_R}.
  \]
\end{proof}

Now we show how Lemmas~\ref{lem:gm} and~\ref{lem:lower} give the
equivalence between statements (1) and (4) of
Theorem~\ref{theo:intro}. Assume (1) and let $c > 0$ be a constant
such that $|Q(D)| \leq |D|^c$ for every $Q \in \mathcal{Q}$ and every
instance $D$. For a fixed join query $Q \in \mathcal{Q}$, if $(x^*_R :
R \in \sigma)$ denotes the optimal fractional edge cover of $Q$,
Lemma~\ref{lem:lower} states that there exist arbitrarily large
instances $D$ such that $|R(D)| = |D|/|\sigma|$ for every $R \in
\sigma$ and
\begin{equation*}
|Q(D)| \geq \prod_{R \in \sigma} |R(D)|^{x^*_R}
\geq (|D|/|Q|)^{\sum_{R \in \sigma} x^*_R} = (|D|/|Q|)^{\rho^*(Q)}.
\label{eq:theotherhalf}
\end{equation*}
In paricular, there exist arbitrarily large instances $D$ for which
$(|D|/|Q|)^{\rho^*(Q)} \leq |D|^c$. It follows that $\rho^*(Q) \leq c$
and hence (4) in Theorem~\ref{theo:intro}.  The converse is even more
direct. Assume (4) and let $c > 0$ be a constant such that $\rho^*(Q)
\leq c$ for every $Q \in \mathcal{Q}$. For a fixed join query $Q \in
\mathcal{Q}$, if $(x^*_R : R \in \sigma)$ denotes the optimal
fractional edge cover of $Q$, Lemma~\ref{lem:gm} states that for every
instance $D$ we have
\begin{equation*}
|Q(D)| \leq \prod_{R \in \sigma} |R(D)|^{x^*_R}
\leq |D|^{\sum_{R \in \sigma} x^*_R} = |D|^{\rho^*(Q)}.
\label{eq:onehalf}
\end{equation*}
It follows that $|Q(D)| \leq |D|^c$ for every $D$ and
hence (1) in Theorem~\ref{theo:intro}.


\subsection{Execution plans}
\label{sec:wc-alg}
It was proved in \cite{gromar06} that there is an algorithm for evaluating a
join query $Q$ in a database $D$ that runs in time
$O\big(|Q|^2\cdot|D|^{\rho^*(Q)+1}\big)$. An analysis of the proof shows that
the algorithms can actually be cast as the evaluation of an explicit (and
simple) join-project plan. For the reader's convenience, we give a proof of this fact here. Combined with the bounds obtained in the previous
section, this yields Theorem~\ref{theo:intro}.

We define the \emph{size} of a $k$-ary relation $R$ to be the number
$||R||:=|R|\cdot k$. The bounds stated in the following fact depend on
the machine model; the statement we give is based on standard random
access machines with a uniform cost measure. Other models may require
additional logarithmic factors.

\begin{fact}\label{fact:jp}
The following hold:
  \begin{enumerate} \itemsep=0pt
  \item The join $R\bowtie S$ of two relations $R$ and $S$ can be computed in
    time $O(||R||+||S||+||R\bowtie S||)$.
  \item The projection $\pi_B(R)$ of an $A$-relation $R$ to a subset
    $B\subseteq A$ can be computed in time $O(||R||)$.
  \end{enumerate}
\end{fact}

For details and a proof of the fact, we refer the reader to
\cite{flufrigro02}.  The following theorem gives the promised
join-project plan:

\begin{theo}\label{theo:alg}
  For every join query $Q$, there is a join-project plan
  for $Q$ that can be evaluated in time
  $O\big(|Q|^2\cdot|D|^{\rho^*(Q)+1}\big)$
  on every given instance $D$. Moreover, there is
  a polynomial-time algorithm that, given $Q$, computes 
  the join-project plan. 
\end{theo}

\begin{proof}
  Let $Q=R_1(A_1)\bowtie\cdots\bowtie R_m(A_m)$ be a join query and
  $D$ an instance for $Q$.  Suppose that the attributes of
  $Q$ are $\{a_1,\ldots,a_n\}$. For
  $i\in[n]$, let $B_i:=\{a_1,\ldots,a_i\}$. Furthermore, let
  \begin{align*}
   \phi_1&:=\big(\cdots(\pi_{B_1}(R_1)\bowtie\pi_{B_1}(R_2))
     \bowtie\cdots\bowtie\pi_{B_1}(R_m)\big),\\
    \phi_{i+1}&:=\Big(\cdots\big((\phi_{i}\bowtie\pi_{B_{i+1}}(R_1))
      \bowtie\pi_{B_{i+1}}(R_2)\big)\bowtie\cdots\bowtie
      \pi_{B_{i+1}}(R_m)\Big)&\text{ for all }i\ge 1.
  \end{align*}
  It is easy to see that for every $i\in[n]$ it holds that
  $\phi_i(D)=\pi_{B_i}(Q(D))$ and hence $\phi_n(D)=Q(D)$. Hence to compute
  $Q(D)$, we can evaluate the join-project plan $\phi_n$.

  To estimate the cost of the evaluating the plan, we need to establish the
  following claim:
  \begin{quote}
    \emph{For every $i\in[n]$ we have
    $|\phi_i(D)|\le|D|^{\rho^*(Q)}$.}
  \end{quote} 
  To see this, we consider the join query
  \[
  Q_i:={R_1}^i\bowtie\cdots\bowtie {R_m}^i,
  \]
  where ${R_j}^i$ is a relation name with attributes $B_i\cap
  A_j$. The
  crucial observation is that $\rho^*(Q_i)\le\rho^*(Q)$,
  because if $(x_R : R \in \sigma)$ is
  fractional edge cover of $Q$, then letting $x_{R^i} = x_R$
  for every $R \in \sigma$ we get a fractional 
  edge cover of $Q_i$ of the same
  cost. If we let $D_i$ be the database instance with
  ${R_j}^i(D_i):=\pi_{B_i}(R_j)$ for all $j\in[m]$, then
  we get
  \[
  \phi_i(D)=Q_i(D_i)\le|D_i|^{\rho^*(Q_i)}\le|D|^{{\rho^*}(Q)}.
  \]
  This proves the claim.

  We further observe that all intermediate results in the computation of
  $\phi_{i+1}(D)$ from $\phi_i(D)$ are contained in
  \[
  \phi_i(D)\times U,
  \]
  where $U$ is the active domain of $D$.
  Hence their size is bounded by
  $|\phi_i(D)|\cdot|D|\le|D|^{{\rho^*}(Q)+1}$,
  and by Fact~\ref{fact:jp} they can be computed in time
  $O(|D|^{{\rho^*}(Q)+1})$. Overall, we have to compute $n\cdot m$
  projections, each requiring time $O(D)$, and $n\cdot m$ joins, each
  requiring time $O(|D|^{{\rho^*}(Q)+1})$. This yields the desired running
  time.
\end{proof}



We shall prove next that join plans perform significantly worse than
join-project plans. Note that to evaluate a join plan one has to
evaluate all its subplans. Hence for every subplan $\psi$ of $\phi$ and every
instance $D$, the size $|\psi(D)|$ is a lower bound for the time required to
evaluate $\phi$ in $D$.

\begin{theo}\label{theo:jp}
  For every $m,N\in\mathbb N$ there are a join query $Q$ and
  an instance $D$ with $|Q|\ge m$ and $|D|\ge N$, and:
  \begin{enumerate} \itemsep=0pt
  \item $\rho^*(Q)\le2$ and hence $|Q(D)|\le |D|^2$ (actually, 
    $|Q(D)|\le|D|$).
  \item Every join plan $\phi$ for $Q$ has a subplan $\psi$ such that
    $|\psi(D)|\ge|D|^{\frac{1}{5}{\log|Q|}}$.
  \end{enumerate}
\end{theo}
\begin{proof}
  Let $n=\binom{2m}{m}$. For every $s \subseteq
  [2m]$ with $|s|=m$, let $a_s$ be an attribute with domain $\mathbb N$.
  For every $i\in[2m]$, let $R_i$ be a
  relation name having as attributes all $a_s$ such that $i
  \in s$. Let $A_i$ be the set of attributes of $R_i$
  and $A = \bigcup_{i \in [2m]} A_i$.
  The arity of $R_i$ is 
  \[
  |A_i| = \binom{2m-1}{m-1}=\frac{m}{2m}\cdot\binom{2m}{m}=\frac{n}{2}.
  \]
  Let $Q:=R_1\bowtie\cdots\bowtie R_{2m}$.
  Then $|Q|=2m\cdot n/2=m\cdot n$.
  Furthermore, $\rho^*(Q)\le 2$. To see this, let
  $x_{R_i} = 1/m$ for every $i \in [2m]$. This forms a
  fractional edge cover of $Q$, because for every
  $s \subseteq [2m]$ with $|s|=m$,
  the attribute $a_s$ appears in the $m$ atoms
  $R_i$ with $i \in s$.
  
  Next, we define an instance $D$ by letting
  $R_i(D)$ be the set of all $A_i$-tuples that have an arbitrary
  value from $[N]$ in one coordinate and $1$ in all other coordinates.
  Formally,
  $$
  R_i(D) := \bigcup_{a \in A_i} \bigcap_{b \in
  A_i\setminus a} \{ t \in \tup(A_i) : t(a) \in [N], t(b)
  = 1 \}.
  $$
  Observe that
  $|R_i(D)|=(N-1)n/2+1$ for all $i\in[2m]$ and thus 
  \[
  |D|=(N-1)mn+2m \geq N.
  \]
  Furthermore, $Q(D)$ is the set of all $A$-tuples that have an
  arbitrary value from $[N]$ in one attribute and $1$ in all other
  coordinates (it is not possible that two attributes have value
  different from 1, as every two attributes appear together in
  some relation).  Hence $|Q(D)|=(N-1)n+1 \le |D|$. This completes the
  proof of (1).

  To prove (2), we shall use the following
  simple (and well-known) combinatorial lemma:

  \begin{lem}
    Let $T$ be a binary tree whose leaves are coloured with $2m$ colours, for
    some $m\ge 1$. Then there exists a node $t$ of $T$ such that at least
    $(m+2)/2$ and at most $m+1$ of the colours appear at leaves that are
    descendants of\/ $t$.
  \end{lem}
  
  \begin{proof}
    For every node $t$ of $T$, let $c(t)$ be the number of colours that appear
    at descendants of $T$. The \emph{height} of a node $t$ is the length of the
    longest path from $t$ to a leaf.

    Let $t$ be a node of minimum height such that $c(t)\ge m+2$, and let
    $u_1,u_2$ be the children of $t$. (Note that $t$ cannot be a leaf because
    $c(t)\ge 2$.) Then $c(u_i)\le m+1$ for $i=1,2$. Furthermore,
    $c(u_1)+c(u_2)\ge c(t)$, hence $c(u_i)\ge{(m+2)/2}$ for at least one $i$.
  \end{proof}

  Continuing the proof of the theorem, we let $\phi$ be a join plan for $Q$.
  We view the term $\phi$ as a binary tree $T$ whose leaves are labelled by
  atoms $R_i$. 
  We view the atoms as colours. Applying the lemma, we find a
  node $t$ of $T$ such that at least $(m+2)/2$ and at most $m+1$
  of the colours appear at leaves that are descendants of $t$. Every inner
  node of the tree corresponds to a subplan of $\phi$. We let $\psi$ be
  the subplan corresponding to $t$. Then at least $(m+2)/2$ and at most
  $m+1$ atoms $R_i$ appear in $\psi$. By symmetry, we may assume
  without loss of generality that the atoms of $\psi$ are $R_1,\ldots,R_\ell$
  for some $\ell\in\big[\ceil{(m+2)/2},m+1\big]$. Hence $\psi$ is a plan
  for the join query
  \[
  R_1\bowtie\cdots\bowtie R_\ell.
  \]
  Let $B :=\bigcup_{i=1}^\ell A_i$ be the set of all attributes
  occurring in $\psi$. For $i\in[m+1]$, let $s_i=\{i\}\cup[m+2,2m]$. Then for
  all $i,j\in[\ell]$ we have $a_{s_i} \in A_j$
  if and only if $i=j$. Hence all
  tuples $t\in \tup(B)$ with $t(a_{s_i})\in[N]$
  for all $i\in[\ell]$ and $t(b)=1$ for all
  $b\in B\setminus\{a_{s_1},\ldots,a_{s_{\ell}}\}$ are contained 
  in $\psi(D)$. As there are $N^\ell$ such tuples, it follows that
  \[
  |\psi(Q)|\ge N^\ell\ge N^{(m+2)/2}.
  \]
  Statement (2) of the lemma follows, because
  \[
  \log|Q|=\log m+\log n\le \log m+\log 2^{2m}=\log m + 2m
  \le 5\cdot (m+2)/2,
  \]
  provided $m$ is large enough, which we may assume
  without loss of generality.
\end{proof}

Statement (2) of the theorem implies that any evaluation algorithm for
the query $Q$ based on evaluating join plans, which may even depend on
the database instance, has a running time at least
$|D|^{\Omega(\log|Q|)}$. This is to be put in contrast with the
running time $O(|Q|^2 \cdot |D|^3)$ from Theorem~\ref{theo:alg}.  It is
a natural question to ask if the difference can be even worse, i.e.,
more than logarithmic in the exponent.

Using the well-known fact that the integrality gap of the linear
program for edge covers is logarithmic in the number of vertices of
the hypergraph (that is, attributes of the join query), we prove below
that for every query $Q$ there is a join plan $\phi$ that can be
evaluated in time $O(|Q|\cdot |D|^{2\rho^*(Q)\cdot\log|Q|)})$, hence
the lower bound is tight up to a small constant factor.
\begin{prop}\label{prop:lower}
  For every join query $Q$, there is a join plan for $Q$ that can be
  evaluated in time $O(|Q|\cdot|D|^{2\rho^*(Q)\cdot\log|Q|)})$
  on every given instance $D$.
\end{prop}

\begin{proof}
  Let $Q$ be a join query with schema $\sigma$. For every $R \in
  \sigma$ let $A_R$ be the set of attributes of $R$ so that $A_\sigma
  = \bigcup_{R \in \sigma} A_R$.  An \emph{edge cover} of $Q$ is a
  subset $\gamma\subseteq\sigma$ such that $A_\sigma \subseteq
  \bigcup_{R \in \gamma} A_R$.  The \emph{edge cover number} $\rho(Q)$
  of $Q$ is the minimum size of an edge cover for $Q$. Observe that
  edge covers correspond to $\{0,1\}$-valued fractional edge covers
  and that the edge cover number is precisely the cost of the optimal
  integral fractional edge cover. It is well known that the
  integrality gap for the linear program defining fractional edge
  covers is $H_n$, where $n=|A_\sigma|$ and $H_n$ is the $n$th
  harmonic number (see, for example, \cite{vaz01}, Chapter 13). It is
  known that $H_n\le 2\log n$. Now the join plan consists in first
  joining the relations that form an edge cover of size
  $2\rho^*(Q)\cdot\log|Q|$ in arbitrary order, and then joining the
  result with the rest of relations in arbitrary order.
\end{proof}

Furthermore, the proof of Proposition~\ref{prop:lower} shows that, for
every join query $Q$, there is a join plan that can be evaluated in
time $O(|Q|\cdot |D|^{\rho(Q)})$, where $\rho(Q)$ denotes the edge
cover number of $Q$. However, note that not only $|D|^{\rho(Q)}$ is
potentially superpolynomial over $|D|^{\rho^*(Q)}$, but also finding
this plan is in general NP-hard.  Compare this with the fact that the
join-project plan given by \cite{gromar06}
can be found efficiently (see Theorem~\ref{theo:alg}).

\section{Size constraints}
\label{sec:constraints}


To estimate the size of joins, practical query optimisers use
statistical information about the database instance such as the sizes
of the relations, the sizes of some of their projections, or
histograms.  In this section we consider the simplest such setting
where the size of the relations is known, and we prove a (worst-case)
estimate on the size of $Q(D)$ subject to the constraint that the
relations in $D$ have the given sizes.

\subsection{Size bounds under size constraints}

Let $Q$ be a join query with schema $\sigma$. For every $R \in
\sigma$, let $A_R$ be the set of attributes of $R$ so that $A_\sigma =
\bigcup_R A_R$.  For every $R \in \sigma$, let $N_R$ be a natural
number, and let $L_Q(N_R : R \in \sigma)$ be the following linear
program:
\begin{equation}
\begin{array}{lll}
  \text{minimise} & \sum_{R} x_R \cdot \log N_R  & \\
  \text{subject to} &
  \sum_{R : a \in A_R} x_R \ge 1 & \text{for all }a\in A_\sigma,\\
  & x_R\ge 0 & \text{for all } R\in \sigma.
\end{array}
\label{eqn:lp2}
\end{equation}
Note that the only difference with $L_Q$ as defined
in (\ref{eqn:lp}) is the objective function.
This implies that every feasible solution of 
$L_Q(N_R : R \in \sigma)$ is also a fractional
edge cover of $Q$.

\begin{theo}\label{theo:size-bounds}
  Let $Q$ be a join query with schema $\sigma$ and 
  let $N_R\in\mathbb N$ for all $R\in\sigma$. Let $n$ be the number of
  attributes of $Q$, and let
  $(x_R : R\in\sigma)$ be an optimal solution of the linear program
  $L_Q(N_R : R\in\sigma)$. 
  \begin{enumerate} \itemsep=0pt
  \item For every $\sigma$-instance $D$ with $|R(D)|=N_R$ for all
  $R$ it holds that
  $|Q(D)|\le\prod_{R}N_R^{x_R}$
  \item There is a $\sigma$-instance $D$ such that $|R(D)|=N_R$ for all
  $R\in\sigma$ and
    $|Q(D)|\ge 2^{-n} \prod_{R}N_R^{x_R}$.
  \end{enumerate}
\end{theo}
\begin{proof}
  Statement (1) is an immediate consequence of Lemma~\ref{lem:gm}.
  To prove (2), we exploit LP duality again.
  The LP-dual of $L_Q(N_R :  R\in\sigma)$ is the following
  linear program $D_Q(N_R : R\in\sigma)$:
  $$
  \begin{array}{lll}
    \text{maximise}& \sum_{a}y_a &\\
    \text{subject to}
    & \sum_{a\in A_R}y_a\le\log N_R &\text{for all }R\in\sigma,\\
    &y_a\ge 0&\text{for all }a\in A_\sigma.
  \end{array}
  $$
  Let $(y_a : a\in A_\sigma)$ be an optimal solution for the
  dual. Then $\sum_{a\in A_\sigma} y_a=\sum_{R\in\sigma} x_R\cdot\log
  N_R$.

  For all $a\in A_\sigma$, let $y'_a=\log\floor{2^{y_a}}\le y_a$.  We
  set
  \[
  R':=\Big\{t \in \tup(A_R) : t(a)\in\big[2^{y'_a}\big]\text{ for all
  }a\in A_R \Big\}.
  \]
  Then
  \[
  |R'|=\prod_{a\in A_R}2^{y_a'}=\prod_{a\in A_R}\floor{2^{y_a}}
  \le 2^{\sum_{a\in A_R}y_a}\le
  2^{\log N_R}=N_R.
  \]
  We arbitrarily add tuples to $R'$ to obtain a relation $R(D)$ of
  size exactly $N_R$. In the resulting instance $D$, we have
  \[
  |Q(D)|\ge \prod_{a\in A_\sigma} 2^{y'_a}\ge\prod_{a\in A_\sigma}
  \frac{2^{y_a}}{2} = 2^{-n} \cdot 2^{\sum_{a\in A_\sigma} y_a}=
  2^{-n} \cdot 2^{\sum_{R\in\sigma}x_R\cdot\log N_R}=2^{-n} \cdot
  \prod_{R\in\sigma}N_R^{x_R}.
  \]
\end{proof}

Even though usually the query is much smaller than the database instance and
 hence we may argue that a constant factor that only depends on the size of the
 query is negligible, the exponential factor in the lower bound of
 Theorem~\ref{theo:size-bounds}(2) is unpleasant. In the following, we shall
 prove that the lower bound cannot be improved substantially. 
In the next example we show that we cannot replace the lower bound of
Theorem~\ref{theo:size-bounds}(2) by
$2^{-(1-\epsilon)n}\prod_{R}N_R^{x_R}$ for
any $\epsilon>0$.  This seems to indicate that maybe the approach to
estimating the size of joins through fractional edge covers is no longer
appropriate in the setting where the size of the relations is fixed. However,
we shall then see that, in some sense, there is no better approach. In
Theorem~\ref{theo:approxhardness}, we shall prove that there is no polynomial
time algorithm that, given a query $Q$ and relation sizes $N_R$, for
$R\in\sigma$, approximates the worst case size of the query answer to a
factor better than $2^{n^{1-\epsilon}}$.

\begin{example}\label{exa:sizelower}
 We give an example where $\prod_{R\in\sigma}N_R^{x_R}$ is roughly
 $2^n$ but $|Q(D)|$ is at most $2^{\epsilon n}$, where
 $n$ is the number of attributes of $Q$. Thus the factor
 $2^{-n}$ in Theorem~\ref{theo:size-bounds}(2) cannot be replaced with
 anything greater than $2^{-(1-\epsilon)n}$.

  Let $n\in\mathbb N$ be an integer, $0<\epsilon<1$ a fixed
  constant, and $A=\{a_1,\ldots,a_n\}$ a set of
  attributes with domain $\mathbb N$. Let $r:=\lfloor \epsilon
  n/\log n\rfloor$. We assume that $n$ is sufficiently large that
  $2^r>n$ holds. 
 For every $B\in\binom{[n]}{r}$,
    let $R_B$ be an $r$-ary relation with attributes $B$. Furthermore,
    for every $a\in A$, let $R_a$ be a unary relation with
    the only attribute $a$.
    Let $Q$ be the
    join of all these relations and
    let $\sigma$ be the resulting schema.
    
      For every $B\in\binom{[n]}{r}$, let $N_{R_B}=2^{r}-1$ and for 
      every $a\in A$, let $N_{R_a}=2$. Consider the linear program
      $L_Q(N_R : R\in\sigma)$. We obtain an optimal solution for this linear
      program by letting
      $x_{R_B}:=n/\left({r\binom{n}{r}}\right)$ and 
      $x_{R_a}:=0$. To see that this is an optimum solution, observe that
       $y_a:=\log (2^r-1)/r$ is a feasible solution of the dual LP with
       the same cost.

      We prove next that $\prod_{R}N_R^{x_R}=2^n(1-o(1))$:
      \[
      \prod_{R\in\sigma}N_R^{x_R}=
      \left({(2^{r}-1)^{n/\left({r\binom{n}{r}}\right)}}\right)^{\binom{n}{r}} \ge
      (2^{r}-1)^{n/r}\ge \left({2^r(1-1/n)}\right)^{n/r}= 
      2^n\left({1-1/n}\right)^{n/r}=2^n(1-o(1)).
      \]
      The second inequality follows from $2^r>n$ and the last equality follows
      from the fact if $n$ tends to infinity, then $(1-1/n)^n$ goes to $1/e$
      and $r$ goes to infinity as well.

      To complete the example, we prove that $|Q(D)|\le 2^{\epsilon
        n}$ for every instance $D$ respecting the constraints
      $N_R$. Let $D$ be a $\sigma$-instance with $|R(D)|=N_R$ for
      every $R\in\sigma$. From $N_{R_a}=2$ it follows that in $Q(D)$
      each attribute has at most two values, hence we can assume
      without loss of generality that $Q(D)\subseteq \{0,1\}^n$. Thus
      each tuple in $t\in Q(D)$ can be viewed as a subset $A_t=\{a\in
      A : t(a)=1\}$ of $A$. For every $B\in \binom{[n]}{r}$, it holds
      $\pi_B(Q(D))\le N_{R_B}=2^r-1$, hence the Vapnik-Chervonenkis
      dimension of $Q(D)$ is less than $r$. Thus by Sauer's Lemma, we
      have
      \[
      |Q(D)|\le n^r \le  n^{\epsilon n/\log n} =2^{\epsilon n},
      \]
      as claimed.
\end{example}

\subsection{Hardness of better approximation}

 There is a gap of $2^n$ between the upper and lower bounds of
 Theorem~\ref{theo:size-bounds}, which means that both  bounds approximate 
 the maximum
 size of $|Q(D)|$ within a factor of $2^n$. However, if $|Q(D)|$ is
 $2^{O(n)}$, then such an approximation is useless. We show that it is
 not possible to find a better approximation in polynomial time: the gap
 between an upper and a lower bound cannot be reduced to 
 $2^{O(n^{1-\epsilon})}$  (under  standard
 complexity-theoretic assumptions). 

 For the following statement, recall that $\textup{ZPP}$ is the class
 of decision problems that can be solved by a probabilistic
 polynomial-time algorithm with \emph{zero-error}. What this means is
 that, on any input, the algorithm outputs the correct answer or
 ``don't know'', but the probability over the random choices of the
 algorithm that the answer is ``don't know'' is bounded by $1/2$.
 Obviously $\textup{P} \subseteq \textup{ZPP} \subseteq \textup{NP}$,
 and the assumption that $\textup{ZPP} \not= \textup{NP}$ is almost as
 believable as $\textup{P} \not= \textup{NP}$ (see
 \cite{PapadimitriouBook}).

\begin{theo}\label{theo:approxhardness}
For a given query $Q$ with schema $\sigma$ and a given set of size 
constraints $(N_R : R\in \sigma)$, denote by $M$ the maximum 
of $|Q(D)|$ over databases
satisfying $|R(D)| = N_R$ for every $R\in\sigma$.
If for some $\epsilon>0$, there is a polynomial-time algorithm that,
given a query $Q$ with $n$ attributes and size constraints $N_R$, 
computes two values $M_L$ and $M_U$ with
$M_L \le M \le M_U$ and
$M_U \le M_L 2^{n^{1-\epsilon}}$,
then $\textup{ZPP} = \textup{NP}$.
\end{theo}
For the proof of Theorem~\ref{theo:approxhardness}, we establish a
connection between the query size and the maximum independent set
problem (Lemma~\ref{lem:indpendentset}).  Then we get our
inapproximability result by reduction from the following result by
H{\aa}stad:

\begin{theo}[\cite{MR1687331}]\label{th:hastad}
If for some $\epsilon_0>0$ there is a polynomial-time algorithm that,
given an $n$-vertex graph $G$, can distinguish between the 
cases $\alpha(G)\le n^{\epsilon_0}$ and
$\alpha(G) \ge n^{1-\epsilon_0}$,
then $\textup{ZPP}=\textup{NP}$.
\end{theo}

Following is the announced connection between worst-case query-size
subject to relation-size constraints and maximum independent sets:

\begin{lem}\label{lem:indpendentset}
 Let $Q$ be a join query with schema $\sigma$ and let $N_R:=2$ for all
 $R\in\sigma$. Let $G$ be the primal graph of $Q$ and let $\alpha(G)$ be
 the size of the maximum independent set in $G$. The maximum  of
 $|Q(D)|$, taken over database instances satisfying $|R(D)|=N_R$ for
 every $R\in \sigma$, is exactly $2^{\alpha(G)}$.
\end{lem}
\begin{proof}
  Let $A_R$ be the attributes of $R \in \sigma$. For this proof we
  write $A$ instead of $A_\sigma$.  First we give a database $D$ with
  $|Q(D)|\ge 2^{\alpha(G)}$. Let $I\subseteq A$ be an independent set
  of size $\alpha(G)$. Since $I$ is independent, $|A_R \cap I|$ is
  either 0 or 1 for every $R\in \sigma$. If $|A_R \cap I|=0$, then we
  define $R(D)$ to contain a tuple that is 0 on every attribute. If
  $A_R \cap I=\{a\}$, then we define $R(D)$ to contain a tuple that is
  0 on every attribute and a tuple that is 1 on $a$ and 0 on every
  attribute in $A_R \setminus \{a\}$. We claim that
\[
Q(D)=\{t\in\tup(A) : t(a)\in \{0,1\} \text{ for all } a\in I,
t(a)=0 \text{ for all } a\in A\setminus I \}.
\]
Clearly, the value of an attribute in $I$ is either 0 or 1, and every
attribute in $A\setminus I$ is forced to 0. Furthermore, any combination of 0 
and 1 on the attributes of $I$ is
allowed as long as all the other attributes are 0. Thus
$|Q(D)|=2^{\alpha(G)}$. Note that a relation $R$ with $|A_R \cap
I|=0$ contains only one tuple in the definition above. To satisfy 
the requirement $|R(D)|=N_R=2$, we
can add an arbitrary tuple to each such relation $R$; this cannot
decrease $|Q(D)|$.

Next we show that if $|R(D)|=2$ for every relation $R\in \sigma$, 
then $|Q(D)|\le 2^{\alpha(G)}$. Since $|R(D)|=2$ for every relation, 
every attribute in $A$ can have at most two values in $Q(D)$; without 
loss of generality it can be assumed that $Q(D)\subseteq
\{0,1\}^{|A|}$. Furthermore, it can be assumed (by a mapping of
the domain of the attributes) that the all-0 tuple is in $Q(D)$.

Let $S$ be the set of those attributes that have two values  in $Q(D)$, i.e.,
\[
S=\{a\in A : |\pi_{\{a\}}(Q(D))|=2\}.
\]
For every $a\in S$, let $S_a$ be the set of those attributes that are
the same as $a$ in every tuple of $Q(D)$, i.e.,
\[
S_a=\{b\in S : t(a)=t(b) \text{ for every } t\in Q(D) \}.
\]

We define a sequence $a_1$, $a_2$, $\dots$ of attributes by letting
$a_i$ be an arbitrary attribute in $S\setminus \bigcup_{j<i}
S_{a_j}$. Let $a_t$ be the last element in this sequence, which means
that $\bigcup_{i=1}^{t}S_{a_i}=S$. We claim that $a_1$, $\dots$, $a_t$
are independent in $G$, implying $t\le \alpha(G)$. Assume that $a_i$
and $a_j$ ($i < j$) are adjacent in $G$; this means that there is an
$R\in \sigma$ with $a_i,a_j\in A_R$. By assumption, the all-0 tuple is
in $R(D)$. As $a_i,a_j\in S$, there has to be a $t_1\in R(D)$ with
$t_1(a_i)=1$ and a $t_2\in R(D)$ with $t_2(a_j)=1$. Since $|R(D)|=2$
and the all-0 tuple is in $R(D)$, we have $t_1=t_2$. But this means
that $a_i$ and $a_j$ have the same value in both tuples in $R(D)$,
implying $a_j\in S_{a_i}$. However, this contradicts the way the
sequence was defined.

Now it is easy to see that $|Q(D)|\le 2^t\le 2^{\alpha(G)}$: by
setting the value of $a_1,\ldots,a_t$, the value of every
attribute in $S$ is uniquely determined and the attributes in $A\setminus
S$ are the same in every tuple of $Q(D)$. 
\end{proof}

\begin{proof}[Proof of Theorem~\ref{theo:approxhardness}]
  We show that if such $M_L$ and $M_U$ could be determined in
  polynomial time, then we would be able to distinguish between the
  two cases of Theorem~\ref{th:hastad}.  Given an $n$-vertex graph
  $G = (V,E)$, we construct a query $Q$ with attributes $V$ and
  schema $\sigma=E$. For each edge $uv\in E$, there is a relation $R_{uv}$
  with attributes $\{u,v\}$. We set $N_R=2$ for every relation
  $R\in\sigma$. Observe that the primal graph of $Q$ is $G$. Thus by
  Lemma~\ref{lem:indpendentset}, $M=2^{\alpha(G)}$.

Set $\epsilon_0:=\epsilon/2$. 
In case (1) of Theorem~\ref{th:hastad}, $\alpha(G)\le n^{\epsilon_0}$, hence
$M_L\le M \le 2^{n^{\epsilon_0}}$ and 
\[
M_U\le M_L 2^{n^{1-\epsilon}} \le 2^{n^{\epsilon_0}+n^{1-\epsilon}} < 
2^{n^{1-\epsilon_0}}
\]
(if $n$ is sufficiently large). On the other hand, in case (2) we have
$\alpha(G)\ge n^{1-\epsilon_0}$, which implies $M_U\ge M = 2^{\alpha(G)}
\geq 2^{n^{1-\epsilon_0}}$. Thus we can distinguish between the two cases
by comparing $M_U$ with $2^{n^{1-\epsilon_0}}$.
\end{proof}


\section{Average-case model}
\label{sec:average}

In this section we assume that the database is randomly generated
according to the following model.  Let $\sigma$ be a schema and let
$A_R$ be the set of attributes of $R \in \sigma$. For every $R \in
\sigma$, let $p_R : \mathbb N \rightarrow (0,1)$ be a function of $N$,
and let $p(N) = (p_R(N) : R \in \sigma)$.  We denote by ${\cal
  D}(N,p(N))$ the probability space on $\sigma$-instances with domain
$[N]$ defined by placing each tuple $t \in [N]^{A_R}$ in $R(D)$ with
probability $p_R(N)$, independently for each tuple $t$ and each $R \in
\sigma$.  Typical probability functions of interest are $p_R(N) =
1/2$, $p_R(N) = C \cdot N^{1-|A_R|}$, or $p_R(N) = N^{1-|A_R|}\log N$.
When $p_R(N) = 1/2$ for every $R \in \sigma$, we are dealing with the
uniform distribution over $\sigma$-instances with domain $[N]$.

\subsection{Size bounds and concentration}
\label{sec:avg-bounds}
Let $Q$ be a join query with schema $\sigma$,
let $n$ be the number of attributes of $Q$,
and let $m$ be the number of relation names in $\sigma$.
Let $X$ denote the size of the query answer $Q(D)$
when $D$ is taken from ${\cal D}(N,p(N))$. The 
expectation of $X$ is, trivially,
\begin{equation}
\Exp[X] = N^n \prod_{R \in \sigma} p_R(N).
\label{eqn:expectation}
\end{equation}
We want to determine under what circumstances 
$|Q(D)|$ is concentrated around this value. For
this we need to compute the variance of $X$,
which depends on a parameter of $Q$ defined
next.

For every $R \in \sigma$, let $w_R$ be a positive real
weight, and let $w = (w_R : R \in \sigma)$. 
The \emph{density} of $Q$ with respect to $w$ is defined as
$\density(Q,w) = \frac{1}{n} \sum_{R \in \sigma} w_R$.
Note that if $w_R = 1$ for every $R$, then
the density is $m/n$.
For every $B \subseteq A_\sigma$, let $Q[B]$ denote
the subquery induced by $B$; that is, $Q[B]$ is the subquery
formed by all the atoms $R \in \sigma$ that have all attributes
in $B$. 
The \emph{maximum density} of $Q$ with respect to $w$ is
$\maxdensity(Q,w) = \max \{ \density(Q[B],w) : B \subseteq A_\sigma, B
\not= \emptyset \}$.

In applications to random instances,
we typically fix $w_R(N)$ to $\log_2(1/p_R(N))$
and write $\density(Q[B])$ and $\maxdensity(Q)$ instead
of $\density(Q[B],w)$ and $\maxdensity(Q,w)$.
For this choice of weights, a crucial distinction is made
according to whether
$\maxdensity$ is larger or smaller than $\log_2(N)$.
In the first case, there exists subquery $Q[B]$ 
whose expected number of solutions is smaller than $1$,
and, therefore, by Markov's inequality, $Q$ itself has no solutions 
at all with probability bounded away from $0$.
In the second case, every subquery has at least
one solution in expectation, and we can bound
the variance of $X$ as a function of $\maxdensity$. 
Since this will be of use later on, we derive it in detail. 
\begin{prop}\label{prop:variance}
If $\maxdensity \leq \log_2(N)$, then
\begin{equation}
\Var[ X ] \leq \Exp[X]^2 \cdot (2^n-1)2^{\maxdensity-\log_2 N}.
\label{eqn:variance}
\end{equation}
\end{prop}
\begin{proof}
For this proof we write $A$ instead of $A_\sigma$.
For every $R \in \sigma$, let $A_R$ be the set of attributes of $R$
and for every $t \in [N]^{A_R}$, let $X(R,t)$ be the indicator
for the event $t \in R(D)$. These are mutually independent random
variables and the expectation of $X(R,t)$ is $p_R(N)$. For every $t \in [N]^A$, 
let $X(t)$ be the indicator for the event $t \in Q(D)$.
Note that $X(t) = \prod_{R \in \sigma} X(R,t_R)$, where $t_R$ denotes 
the projection of $t$ to the attributes of $R$. Also
$X = \sum_t X(t)$. Towards proving (\ref{eqn:variance}),
let us bound
\begin{align}
\Exp\left[{X^2}\right] = \sum_{s,t} \Exp\big[{X(s) X(t)}\big].
\end{align}
For every fixed $B \subseteq A$, let $F_B$ be
the set of pairs $(s,t) \in [N]^A \times [N]^A$ 
such that $s(a)=t(a)$ for every $a \in B$ and
$s(a) \not= t(a)$ for every $a \in A-B$. Clearly,
$(F_B)_{B \subseteq A}$ is a partition of $[N]^A \times [N]^A$
and therefore
\begin{equation}
\sum_{s,t} \Exp\big[X(s) X(t)\big] =
\sum_{B \subseteq A} \sum_{(s,t) \in F_B}
\Exp\big[X(s) X(t)\big].
\end{equation}
Now fix some $B \subseteq A$ and $(s,t) \in F_B$,
and let $\sigma_B$ be the relations appearing
in $Q[B]$. Observe that since $s$ and $t$
agree on $B$ we have $t_R = s_R$ for
every $R \in \sigma_B$ and therefore $X(R,s_R) X(R,t_R) =
X(R,s_R)$ for every such $R$. Hence:
\begin{equation}
X(s) X(t) = \prod_{R \in \sigma} X(R,s_R)
\prod_{R \in \sigma} X(R,t_R) =
\prod_{R \in \sigma-\sigma_B} X(R,s_R) X(R,t_R)
\prod_{R \in \sigma_B} X(R,s_R).
\end{equation}
All variables in the right-hand side product
are mutually independent because either they involve different
relations or different tuples. Therefore,
\begin{equation}
\Exp\big[X(s) X(t)\big] = \prod_{R \in \sigma-\sigma_B} p_R^2
\prod_{R \in \sigma_B} p_R = \prod_{R \in \sigma} p_R^2
\prod_{R \in \sigma_B} p_R^{-1}.
\end{equation}
The number of pairs $(s,t)$ in $F_B$ is bounded 
by $N^{2|A|-|B|}$. Therefore,
\begin{equation}
\sum_{(s,t) \in F_B} \Exp\big[X(s) X(t)\big] \leq
N^{2|A|-|B|}\prod_{R\in\sigma} p_R^2 
\prod_{R\in\sigma_B}p_R^{-1} =
\Exp\big[X\big]^2 \cdot N^{-|B|} \prod_{R \in \sigma_B} p_R^{-1}.
\label{eqn:prod}
\end{equation}
For $B = \emptyset$, the second factor in the right-hand side
of (\ref{eqn:prod}) is $1$ and
we get $\Exp\big[X\big]^2$. For
$B \not= \emptyset$, we have
\begin{equation}
N^{-|B|} \prod_{R\in\sigma_B} p_R^{-1} =
N^{-|B|} 2^{|B|\density(Q[B])} \leq \left({N^{-1} 2^{\maxdensity}}
\right)^{|B|} \leq 2^{\maxdensity-\log_2 N}
\end{equation}
where the first inequality holds 
because $\density(Q[B]) \leq \maxdensity$,
and the second inequality holds because $|B| > 0$
and $\maxdensity \leq \log_2(N)$.
Putting it all together we get
\begin{equation}
\Exp\big[X^2\big] = \sum_{B \subseteq A} \sum_{(s,t) \in F_B}
\Exp\big[X(s) X(t)\big] \leq \Exp\big[X\big]^2 + (2^n-1)\Exp\big[X\big]^2 
2^{\maxdensity-\log_2 N}.
\end{equation}
Since $\Var\big[X\big] = \Exp\big[X^2\big] - \Exp\big[X\big]^2$,
this proves (\ref{eqn:variance}).
\end{proof}

In the following, if $X$ is a random variable
defined on the probability space ${\cal D}(N,p(N))$,
the expression
``$X \sim x$ almost surely'' means that
for every $\epsilon > 0$ and $\delta > 0$, there
exists $N_0$ such that,
for every $N \geq N_0$, we have
$\Prob[ |X - x| \leq \epsilon x ]
\geq 1-\delta$.
With all this notation, we obtain the following
threshold behaviour as an immediate consequence
to Markov's and Chebyshev's inequalities:

\begin{theo} \label{thm:avgcase}
Let $Q$ be a join query with schema $\sigma$ and $n$
attributes.
For every $R \in \sigma$, let $p_R : \mathbb N \to (0,1)$,
$p(N) = (p_R(N) : R \in \sigma)$, 
and $\maxdensity(N) = \maxdensity(Q,w_R(N))$ for
$w_R(N) = \log_2(1/p_R(N))$.
Let $D$ be drawn from ${\cal D}(N,p(N))$ and let $X$ denote
the size of $Q(D)$.
\begin{enumerate} \itemsep=0pt
\item If $\maxdensity(N) = \log N - \omega(1)$,
then $X \sim N^n \prod_{R \in \sigma} p_R(N)$ almost surely.
\item If $\maxdensity(N) = \log N + \omega(1)$, then $X = 0$ almost
surely.
\end{enumerate}
\end{theo}

\begin{proof}
For this proof we write $A$ instead of $A_\sigma$.
We start with (2). Suppose that $\maxdensity(N) = \log N +
\omega(1)$ and fix a large $N$.
Let $B \subseteq A$, $B \not= \emptyset$, be
such that $\maxdensity(Q,w(N)) = \density(Q[B],w(N))$.
Let $Q_B = Q[B]$, let $\sigma_B$ be the
schema of $Q_B$, and let $M_B = |Q_B(D)|$.
The expectation of $M_B$ is
\[
N^{|B|} \prod_{R \in \sigma_B} p_{R}(N)=2^{|B|(\log N-|B|^{-1}\sum_{R \in \sigma_B} \log(1/p_R(N)))}
=2^{|B|(\log N-\delta(Q_B))}.
\]
Since $\density(Q_B) = \maxdensity(Q)$ and
$|B| > 0$, the hypothesis $\maxdensity(N) = \log N
+ \omega(1)$ implies that this quantity approaches
$0$ as $N$ grows. By Markov's inequality, $M_B=0$
almost surely, and therefore $M=0$ almost surely
because every solution to $Q$ gives a solution to $Q_B$.

For (1) we use the bound on the variance from (\ref{prop:variance}).
Fix $\epsilon > 0$ and $\delta > 0$.
By Chebyshev's inequality we have
\[
\Prob[ | X - \Exp[X] | \geq \epsilon \Exp[X] ] \leq
\frac{\Var[X]}{\epsilon^2 \Exp[X]^2} \leq \frac{\Exp[X]^2 (2^n-1) 2^{\maxdensity-
\log_2 N}}{\epsilon^2 \Exp[X]^2} \leq \frac{2^n-1}{\epsilon^2}\cdot 2^{\maxdensity-
\log_2 N}.
\]
Under the hypothesis $\maxdensity(N) = \log N - \omega(1)$, the right-hand side
is bounded by $\delta$ for large enough $N$. Since $\Exp[X] = N^n \prod_{R \in \sigma} p_R(N)$,
the result follows.
\end{proof}

In certain applications, the concentration defined
by ``$X \sim x$ almost surely'' is not enough. For example, it may 
sometimes
be necessary to conclude that $\Prob[|X - x| \leq \epsilon x]
\geq 1-N^{-d}$ for every $\epsilon > 0$ and $d > 0$
in order to apply a union bound
that involves a number of cases that grows polynomially with $N$.
Accordingly, for a random variable $X$, the expression ``$X \sim x$
polynomially almost surely'' means that for every $\epsilon > 0$
and $d > 0$, there exists $N_0$ such that for
every $N \geq N_0$ we have $\Prob\big[|X-x| \leq \epsilon x]
\geq 1-N^{-d}$. 

It turns out that such a strong concentration can also
be guaranteed at the expense of a wider  \emph{threshold width}
in Theorem~\ref{thm:avgcase}: instead of $\log_2(N) - \omega(1)$
vs $\log_2(N) + \omega(1)$, we require $\log_2(N) - \omega(\log\log(N))$ 
vs $\log_2(N) + \omega(1)$. This does not follow from
Chebyshev's inequality, and for the proof we use the polynomial concentration 
inequality from
\cite{KimVu}.

\begin{theo} \label{thm:avgcase2}
Let $Q$ be a join query with schema $\sigma$ and $n$
attributes.
For every $R \in \sigma$, let $p_R : \mathbb N \to (0,1)$,
$p(N) = (p_R(N) : R \in \sigma)$, 
and $\maxdensity(N) = \maxdensity(Q,w_R(N))$ for
$w_R(N) = \log_2(1/p_R(N))$.
Let $D$ be drawn from ${\cal D}(N,p(N))$ and let $X$ denote
the size of $Q(D)$.
\begin{enumerate} \itemsep=0pt
\item If $\maxdensity(N) = \log N - \omega(\log\log N)$,
then $X \sim N^n \prod_{R \in \sigma} p_R(N)$ polynomially almost surely.
\item If $\maxdensity(N) = \log N + \omega(1)$, then $X = 0$ almost
surely.
\end{enumerate}
\end{theo}

For the proof, we will use the polynomial
concentration method from \cite{KimVu}.
Let $H = (V,E)$ be a hypergraph with $n = |V|$
and $k = \max_{e \in E} |e|$. For every $e \in E$,
let $w(e)$ be a positive weight.
Let $\{ X_u : u \in V \}$ be a collection of mutually
independent random variables, where each $X_u$ is
an indicator random variable with expected
value $p_u$. 
Here $0 \leq p_u \leq 1$ for every $u \in V$. Let
$M$ be the following polynomial: $M = \sum_{e \in E} w(e) 
\prod_{u \in e} X_u$. For every $Y \subseteq V$, let $M_Y$ be 
the partial derivative of $M$ with respect to $\{ X_u : u
\in Y \}$, that is, $M_Y=\sum_{e\in E:\;Y\subseteq e}w(e)\prod_{u\in
  e\setminus Y}X_u$.
Let $E_Y = \Exp[M_Y]$, and for every $i \in \{0,\ldots,k\}$,
let $E_i = \max \{ E_Y : Y
\subseteq V, |Y| = i \}$. Note that $E_0 = \Exp[M]$.
Let $E' = \max \{ E_i : 1 \leq i \leq k\}$ and
$E = \max \{ E_0,E' \}$.

\begin{theo}[Theorem 7.8.1 in \cite{AlonSpencerBook}]
\label{thm:concentration}
For every $\lambda > 1$, it holds 
\[
\Prob\left[{|M-\Exp[M]| > a_k (EE')^{1/2} \lambda^k}
\right] < d_k e^{-\lambda} n^{k-1},
\]
where $a_k = 8^k \sqrt{k!}$ and $d_k = 2e^2$.
\end{theo}

For a complete proof of Theorem~\ref{thm:avgcase2}, it
is easier to first state the following consequence 
of Theorem~\ref{thm:concentration} (see
also Corollary 4.1.3 in \cite{KimVu}):
 
\begin{cor} \label{cor:kimvu}
For every two reals $\epsilon > 0$ and $d > 0$, 
and every integer $k \geq 1$, there exists $n_0$
such that, in the setting of
Theorem~\ref{thm:concentration}, 
if $n \geq n_0$ and $E_i/E_0 \leq (\log n)^{-4k}$
for every $i \in \{1,\ldots,k\}$, then
$$
\Prob\left[{ |M-\Exp[M]| > \epsilon \Exp[M] }\right]
\leq n^{-d}
$$
\end{cor}

For a proof, it suffices to choose
$\lambda(n) = (\epsilon/a_k)^{1/k} (\log n)^{2}$
and $n_0$ such that $\lambda > 1$ and
$d_k e^{-\lambda} n^{k-1} < n^{-d}$
for every $n \geq n_0$. Note that
if $E_i/E_0 \leq (\log n)^{-4k}$ then $E = E_0$ and
$E' \leq E_0 (\log n)^{-4k}$, so $a_k(EE')^{1/2}\lambda^k
\leq \epsilon E_0$ and the claim follows.

\begin{proof}[Proof of Theorem~\ref{thm:avgcase2}]
The proof of (2) is the same as in Theorem~\ref{thm:avgcase}.
For (1), suppose that $\maxdensity(N) = \log N -
\omega(\log\log N)$ and fix a large $N$.
For every $R \in \sigma$, let $A_R$ be the
set of attributes of $R$, and for every $t \in [N]^{A_R}$,
let $X(R,t)$ be the indicator random variable
for the event $t \in R(D)$. Note that these are
mutually independent random variables and
$\Exp[X(R,t)] = p_R(N)$ by the
definition of the probability space.
Note also that
$$
M = \sum_{t} \prod_{R \in \sigma} X(R,t_R),
$$
where $t$ ranges over all tuples in $[N]^A$,
and $t_R$ denotes the projection of $t$ to
the attributes of $R$. We aim for an application
of Corollary \ref{cor:kimvu}
with the random variables $X(R,t)$. We define
the hypergraph $H = (V,E)$. The set of vertices $V$
is the set of pairs $(R,t)$ where
$R \in \sigma$ and $t \in [N]^{A_R}$. There
is one hyperedge $e_t$ in $E$ for every $t \in [N]^A$
that consists of all pairs $(R,t_R)$ with $R \in \sigma$.
Thus, the number of vertices is bounded by $mN^r$,
where $m = |\sigma|$ and $r$ is the maximum arity
of the relations in $\sigma$. Furthermore, the maximum size of the edges in
$H$, that is, the $k$ in Corollary~\ref{cor:kimvu}, is $m$.

We have
\begin{equation}
E_0 = N^n \prod_{R \in \sigma} p_R(N). \label{eqn:e0}
\end{equation}
Let us bound $E_i$ for $i > 0$.
Fix a set of vertices of the hypergraph $H$, say
$$
Y \subseteq \{ (R,t) : R \in \sigma, t \in [N]^{r_R} \},
$$
with $|Y|=i > 0$. We want to bound $E_Y/E_0$.
Let $\sigma'$ be the set of relations $R$ that appear in
$Y$. Let $B$ be all the attributes of the relations in
$\sigma'$. Note that $\sigma' \subseteq \sigma_{B}$,
where $\sigma_{B}$ is the schema of $Q[B]$. This will be
of use later.

Let $T$ be the set of all $t \in [N]^A$ such
that $(R,t_R) \in Y$ for every $R \in \sigma'$,
where as before, $t_R$ denotes the projection of $t$
to the attributes of $R$.
If there exist $t_1$ and $t_2$ in $T$ that disagree
on some attribute of $B$, then automatically $E_Y = 0$
because then $Y$ is not included in any hyperedge $e_t$
of $H$ and hence $M_Y=0$. We may assume then that all 
$t \in T$ agree on $B$. This implies $|T| \leq N^{n-|B|}$. Under
these conditions we have
$$
E_Y = \Exp\left[ { \sum_{t \in T} \prod_{R \in
\sigma-\sigma'} X(R,t_R)
} \right] \leq N^{n-|B|} \prod_{R \in \sigma-\sigma'} p_R(N).
$$
Therefore, recalling (\ref{eqn:e0}), we bound
$E_Y/E_0$ by
\begin{equation}
N^{-|B|} \prod_{R \in \sigma'} \frac{1}{p_R(N)} \leq N^{-|B|}
\prod_{R \in \sigma_B} \frac{1}{p_R(N)} \leq N^{-|B|\left({1-\frac{1}{\log N}
|B|^{-1}\sum_{R \in \sigma_B} \log(1/p_R(N))}\right)},
\label{eqn:thing}
\end{equation}
where the first inequality holds because
$\sigma' \subseteq \sigma_{B}$ and each $p_R$ belongs to
$(0,1)$.
Using the hypothesis that $\maxdensity(N) = \log N -
\omega(\log\log N)$, we bound (\ref{eqn:thing})
by $(\log N)^{-(4m+1)}$.
We showed then that
$$
E_Y/E_0 \leq (\log N)^{-(4m+1)},
$$
and since this holds for an arbitrary $Y$ of size $i > 0$,
the bound is also valid for $E_i/E_0$. Recall now
that the number of vertices $|V|$ of the hypergraph $H$
is at most $mN^r$, and we can bound
$$
(\log N)^{-(4m+1)} \leq (\log (mN^r))^{-4m} \leq (\log
|V|)^{-4m}
$$
for large $N$. The result follows
from Corollary~\ref{cor:kimvu}.
\end{proof}

We conclude this section with a max-flow construction to compute
$\maxdensity(Q,w)$. This is probably folklore, but as the proof is
short we include it anyway for the reader's convenience. For every
real number $\delta > 0$, we build a network $N(\delta)$ as
follows. The network has a source $s$, a target $t$, and
$|A_\sigma|+|\sigma|$ intermediate nodes. There is a link of capacity
$\delta$ between $s$ and each $a \in A_\sigma$. Each $a \in A_\sigma$
has a link of infinite capacity to each $R \in \sigma$ with $a \in
A_R$. Finally, each $R \in \sigma$ is linked to $t$ with capacity
$w_R$.  Recall that a cut in the network is a set of links that
disconnects the target from the source. The capacity of the cut is the
sum of the capacities of the links in it. Let $\gamma(Q,w,\delta)$ be
the minimum capacity of all cuts of $N(\delta)$.
\begin{lem} \label{lem:maxflow}
The following are equivalent:
\begin{enumerate}\itemsep=0pt
\item $\gamma(Q,w,\delta) < \sum_{R \in \sigma} w_R$
\item $\maxdensity(Q,w) > \delta$.
\end{enumerate}
\end{lem}
\begin{proof}
Assume $\maxdensity(Q,w) > \delta$ and let $B \subseteq A_\sigma$,
$B \not= \emptyset$, be such that $\density(Q[B],w) > \delta$.
Let $\sigma_B$ be the schema of $Q[B]$.
Let $S$ be the cut that consists of all links from the source
to the nodes of $B$ and the links from the nodes in
$\sigma-\sigma_B$ to the target. The capacity of this cut
is 
$$
|B|\delta+\sum_{R \in \sigma-\sigma_B} w_R
< \sum_{R \in \sigma_B} w_R + \sum_{R \in \sigma-\sigma_B}
w_R = \sum_{R \in \sigma} w_R,
$$
where the inequality follows from $\density(Q[B],w) > \delta$.

Suppose now $\gamma(Q,w,\delta) < \sum_R w_R$ and let $S$ be 
a cut of minimum capacity.
Let $B$ be the set of $a \in A_\sigma$ for which the link
from the source to $a$ is in $S$. Let $\sigma_B$
be the schema of $Q[B]$. We claim that $S$ does
not contain any link from an $R \in \sigma_B$ to the
target. For if it did, $S-\{(R,t)\}$ would also be a cut
of smaller capacity. Also $S$ contains all links
from an $R \in \sigma-\sigma_B$ to $t$. For if it
did not, $S$ would not be a cut (we assume
all $R$ have a non-empty set of attributes). Finally, $S$ does not 
contain any link from an $a \in A_\sigma$ to an $R \in \sigma$ because
those have infinite capacity. Therefore, the capacity of $S$
is
$\delta|B|+\sum_{R \in \sigma-\sigma_B} w_R$
and smaller than $\sum_{R \in \sigma} w_R$ by
hypothesis. Hence $B \not= \emptyset$ and $\density(Q[B],w) > \delta$.
\end{proof}
By the max-flow min-cut algorithm, it follows that 
$\maxdensity(Q,w)$ is computable in polynomial time.
 
\subsection{Execution plans}\label{sec:avg-alg}
Theorem \ref{theo:jp} shows that certain queries admit a join-project
plan that cannot be converted into a join plan without causing a
superpolynomial increase in the worst-case running time. The following
result shows that when we are considering average-case running time in
the random database model, projections may be eliminated at a very
small expected cost.


\begin{theo} \label{thm:removeproject}
Let $Q$ be a join query with schema $\sigma$ and $n$
attributes. Let $\phi$ be a join-project plan for $Q$.
For every $R \in \sigma$, let $p_R : \mathbb N \to (0,1)$,
$p(N) = (p_R(N) : R \in \sigma)$ and let $D_N$ be drawn
from ${\cal D}(N,p(N))$. There exists a constant $c_\phi$ 
depending only $\phi$, such that for every large enough $N$
and every $T$, if $\Exp[ |\psi(D_N)| ] \leq T$ for every
subplan $\psi$ of $\phi$, then there is a join plan $\phi^*_N$
for $Q$ such that $E[ |\psi(D_N)| ] \leq c_\phi T$ for
every subplan $\psi$ of $\phi^*_N$.
\end{theo}

That is, informally speaking, for every join-project plan there is a
join plan that is ``almost as good''. Note that in
Theorem~\ref{thm:removeproject}, the join plan depends on $N$. This is
an unavoidable artifact of our model: the probabilities $p_R(N)$ can
be completely different for different values of $N$, hence it is
unavoidable that different plans could be need for different
$N$.

The join plan $\phi^*$ is obtained by iteratively using a procedure
that is capable of reducing the number of projections by one in such a
way that the expected size of each subplan increases only by a factor
depending only on the query. In each iteration, the procedure selects
a subplan $\pi_A(\phi_0)$ of $\phi$ such that $\phi_0$ contains no
projections, i.e. this projection $\pi_A$ is lowest in the tree
representation of $\phi$. In the first step of the procedure, we
replace $\phi_0$ with a join plan $\phi'_0$ that contains only those
relations appearing in $\phi_0$ whose attributes are completely
contained in $A^*$, where $A^*\supseteq A$ is an appropriate set of
attributes from $A_\sigma$. In the second step, the projection $\pi_A$
is removed (or, in other words, $\pi_A$ is replaced by $\pi_{A^*}$,
making it redundant). The key step of the algorithm is choosing the
right $A^*$. If $A^*$ is too small, then $\phi'_0$ is much less
restrictive than $\phi_0$, hence $|\pi_A(\phi'_0(D))|$ can be much
larger than $|\pi_A(\phi_0(D))|$. On the other hand, if $A^*$ is too
large, then $|\pi_{A^*}(\phi'_0(D))|$ can be much larger than
$|\pi_A(\phi'_0(D))|$. The algorithm carefully balances the size of
$A^*$ between these two opposing constraints; the choice of $A^*$ is
based on the minimization of a submodular function defined below.

Let $S\subseteq \sigma$ be a set of relations over the attributes
$A_\sigma$ and denote by $A_R$ the attributes of a relation $R$. For a
subset $A \subseteq A_\sigma$, let
\[
f_S(A):=|A|(\log N-n-1)-\sum_{R\in S[A]}w_R(N),
\]
where the set $S[A]$ contains those relations $R\in S$ whose
attributes are contained in $A$.  It is easy to see that $f_S(A)$ is
submodular, i.e.,
\[
f_S(A)+f_S(B)\ge f_S(A\cup B)+f_S(A\cap B)
\]
for every $A,B\subseteq A_\sigma$. It follows that $A$ has a unique
minimum-value extension:
\begin{prop}\label{prop:averageclosure}
  For every $A\subseteq A_\sigma$ and $S\subseteq \sigma$, there is a
  unique $C_S(A)\supseteq A$ such that $f_S(C_S(A))$ is minimal and,
  among such sets, $|C_S(A)|$ is maximal.
\end{prop}
\begin{proof}
  Suppose that there are two such sets $B$ and $C$ with this property.
  By the minimality of $f_S(B)=f_S(C)$, we have $f_S(B\cup C)\ge
  f_S(B)$ and $f_S(B\cap C)\ge f_s(C)$. Furthermore, as $|B\cap
  C|<|C|$, the maximality of $|C|$ among the sets minimizing $f_S$
  ensures that $f_S(B\cap C)$ is strictly greater than $f_S(C)$. It
  follows that $f_S(B\cup C)+f_S(B\cap C)>f_S(B)+f_S(C)$, violating
  the submodularity of $f_S$.
\end{proof}





We prove Theorem~\ref{thm:removeproject} by presenting an algorithm
that iteratively removes projections. We describe the algorithm below,
then we prove in Sections~\ref{sec:probtools}-\ref{sec:step2} that the
algorithm transforms the execution plan the required way.
\medskip

\textbf{The algorithm.}  First, we can assume that $\phi$ is of the
form $(((\phi'\bowtie R_1)\bowtie R_2)\bowtie \dots \bowtie
R_{|\sigma|})$, where $R_1$, $\dots$, $R_{|\sigma|}$ is an ordering of
the relations in $\sigma$: if $\phi'$ is a join-project plan for the
query $Q$, then joining any relation $R_i$ with $\phi'$ does not
change $\phi'$. However, this assumption will ensure that if we make
any changes in $\phi'$, then $\phi$ will remain a join-project plan
for the query $Q$. Furthermore, we assume that for every attribute
$a\in A_\sigma$, there is a dummy unary relation $R_a$ with
$p_{R_a}(N)=1$ and hence $w_{R_a}(N)=0$. (We can always remove joins
with these dummy relations from our final join plan without increasing
the size of any intermediate joins.)

The two steps described below reduce the number of projections in
$\phi$ in such a way that the maximum expected size of a subplan is at
most a constant factor larger in the new plan than in $\phi$ (with a
constant depending only on $\phi$). This procedure is repeated as many
times as the number of projections, thus the total increase of the
maximum expected size is only a constant $c_\phi$.

Let $\pi_A(\phi_0)$ be a subplan of $\phi$ such that $\phi_0$
does not contain any projections. Let $S\subseteq \sigma$ be the
relation names appearing in $\phi_0$, which means that
$\phi_0(D)=\bowtie_{R\in S} R(D)$. Let $A^*=C_S(A)$.

\begin{itemize}
\item \textbf{Step 1 (removing joins).}  If a relation $R\in
  S\setminus S[A^*]$ appears in $\phi_0$, then $R$ is removed from
  $\phi_0$.  For the sake of analysis, we implement the removal by
  replacing the relation $R$ in $\phi_0$ with the 0-ary relation (note
  that joining the 0-ary relation with any relation $R'$ gives exactly
  $R'$). This way, it will be clear that there is a correspondence
  between the subplans of the original and modified execution
  plans. We can assume that for every $a\in A^*$, the unary relation
  $R_a$ appears in $\phi_0$ (such relations can be joined with
  $\phi_0$ without increasing the number of tuples).  This ensures
  that $A^*\subseteq A_{\phi'_0}$.  Replacing subplan $\phi_0$ of
  $\phi$ with $\phi'_0$ gives a new join-project plan $\phi'$. By our
  initial assumption on the structure of $\phi$ (that all the
  relations are joined just below the root), $\phi'$ is also a
  join-project plan for $Q$.

\item \textbf{Step 2 (removing projections).} In the second step of
  the procedure, we obtain a join-project plan $\phi''$ from $\phi'$
  by replacing the subplan $\pi_A(\phi'_0)$ with $\pi_{A^*}(\phi'_0)$.
  (This makes the projection redundant and can be eliminated, but it
  is more convenient to analyze the step this way, since this step
  does not change the structure of the join-project plan.) Note that
  this change does not have any effect on the size of $\psi'(D)$ for
  any subplan $\psi'$ of $\phi'_0$.
\end{itemize}
See Figure~\ref{fig:plan} for an example.  It Sections~\ref{sec:step1}
and \ref{sec:step2}, we show that the two steps increase the expected
size of the query only by a constant factor. In
Section~\ref{sec:probtools}, we review and introduce the probabilistic
tools that are required for this analysis.

\begin{figure}
\begin{align*}
R_1:& \{a,c\}\\
R_2:& \{b,c\}\\
R_3:& \{c,e\}\\
R_4:& \{b,d\}\\
R_5:& \{d,f\}\\
R_6:& \{e,f\}\\
\end{align*}
{\automath \Tree 
[.$\bowtie$ 
[.$\bowtie$
[.$\bowtie$ 
[.$\bowtie$ 
[.$\bowtie$ 
[.$\bowtie$ 
[.$\pi_{\{a,c,f\}}$
[.$\bowtie$ $R_2$
[.$\bowtie$ [.$\pi_{\{d,f\}}$ [.$\bowtie$ [ $R_5$ $R_6$ ] ] ]
[.\fbox{$\pi_{\{a,b\}}$} [.$\bowtie$ [.$\bowtie$
[.$\bowtie$ [ $R_1$ \fbox{$R_3$} ] ]
[.$\bowtie$ [ $R_2$ $R_4$ ] ] ] \fbox{$R_5$} ] ] ] ] ]
$R_1$ ] $R_2$ ] $R_3$ ] $R_4$ ] $R_5$ ] $R_6$ ]
}
\caption{A join-project plan over a schema of 6 relations
  $\{R_1,R_2,R_3,R_4,R_5,R_6\}$, represented as a binary
  tree. We demonstrate the removal of a projection from the join-project plan in the proof of Theorem~\ref{thm:removeproject}. Consider the 
  two steps when removing the projection $\pi_{\{a,b\}}$.  Let us assume that $C_S(\{a,b\})=\{a,b,c,d\}$. In Step 1, the two framed
  relations $R_3$ and $R_5$ are removed. In Step 2, $\pi_{\{a,b\}}$ is replaced by
  $\pi_{\{a,b,c,d\}}$.  }

\label{fig:plan}
\end{figure}

\subsubsection{Probabilistic tools}
\label{sec:probtools}
The FKG Inequality is a general tool for determining the correlation
between monotone (antimonotone) events:

\begin{fact}\label{fact:fkg}
  Let $V$ be a finite set, and let $f,g: \{0,1\}^V\to
  \mathbb{R}$ be monotone functions on these variables. Let $\mu: \{0,1\}^V\to
  \mathbb{R}^+$ be a function satisfying $\mu(x)\mu(y)\le \mu(x\vee
  y)\mu(x\wedge y)$ for every $x,y\in \{0,1\}^V$ (where $x\vee y$ and $x\wedge
  y$ denote the coordinate-wise disjunction and conjunction of the two tuples,
  respectively.) Then
\[
\left( \sum_{x\in \{0,1\}^V}f(x)g(x)\mu(x)\right)
\left( \sum_{y\in \{0,1\}^V}\mu(y)\right)\ge
\left( \sum_{x\in \{0,1\}^V}f(x)\mu(x)\right)
\left( \sum_{y\in \{0,1\}^V}g(y)\mu(y)\right).
\]
\end{fact}
A proof can be found in \cite{AlonSpencerBook}.

In general, if $f$, $g$, $h$ are three monotone 0-1 functions of a set
of independent 0-1 random variables, then (somewhat
counterintuitively) it is not necessarily true that $\Prob[ f=1 \;|\;
g=1 ] \le \Prob[ f=1 \;|\; gh=1 ]$. That is, the condition that a more
restrictive monotone function ($gh$ instead of $g$) is 1 does not
necessarily increase the probability that the monotone function is 1.
As an example, suppose that $x_1$, $x_2$, $x_3$ are independent 0-1
random variables, each having probability $1/2$ of being 1. Let
\begin{align*}
f&=(x_1\wedge x_2)\vee (x_2\wedge x_3)\vee (x_1\wedge x_3)\\
g&=(x_1\wedge x_2)\vee x_3\\
h&=x_3
\end{align*}
Now $\Prob[ f=1\;|\;g=1 ] =4/5$, but $\Prob[ f=1\;|\;gh=1 ] = \Prob[
f=1\;|\;h=1 ] =3/4$.  However, the statement is true in the special
case when the $g$ and $h$ are products of the random variables:
\begin{lem}\label{lem:poscor}
  Let $V$ be a set of independent 0-1 random variables, let $M'\subseteq M
  \subseteq V$ be two
  subsets of these variables, and let $f:\{0,1\}^V\to \{0,1\}$ be a
  monotone   function. Then 
  $$
  \Prob[ f=1 \;|\; \textstyle{\prod} M=1 ] \ge \Prob[ f=1 \;|\; 
  \textstyle{\prod} M'=1 ].
  $$
\end{lem}
\begin{proof}
  For every $x\in \{0,1\}^V$, if every variable of $M'$ is 1 in $x$,
  then let $\mu(x)$ be the probability of tuple $x$, otherwise let
  $\mu(x)=0$. It is easy to verify that $\mu(x)\mu(y)\le \mu(x\vee
  y)\mu(x\wedge y)$ for every $x,y\in \{0,1\}^V$. With $f$ and
  $g=\prod (M\setminus M')$, Theorem~\ref{fact:fkg} implies
\begin{align*}
\Prob[ f=1 \wedge \textstyle{\prod}  M=1 ]\cdot \Prob[ 
\textstyle{\prod} M'=1 ] &\ge 
\Prob[ f=1 \wedge \textstyle{\prod} M'=1 ] \cdot 
\Prob[ \textstyle{\prod} M=1 ]
\end{align*}
which, rewritten, is
\begin{align*}
\Prob[ f=1 \;|\; \textstyle{\prod} M=1 ]  &\ge 
\Prob[ f=1 \;|\; \textstyle{\prod} M'=1 ],
\end{align*}
what we had to show.
\end{proof}

The choice of the random database $D$ can be thought of as a set of
independent 0-1 random variables ($N^r$ variables describe an $r$-ary random
relation). For a join-project plan $\phi$ and a tuple $t$, let
$I_{\phi(D),t}$ be the indicator random variable that is 1 if and only
if $t\in \phi(D)$; clearly $I_{\phi(D),t}$ is a monotone function of
the random variables describing the database. Since this function is
monotone, it can be expressed as the disjunction of minterms, i.e.,
$I_{\phi(D),t}= \bigvee_{i=1}^M I^{(i)}_{\phi(D),t}$, where each
minterm $I^{(i)}_{\phi(D),t}$ is the product of a subset of the 0-1
random variables. We say that the {\em rank} of a monotone function is
the maximum size of a minterm of the function. The following two
lemmas will be useful for determining conditional probabilities
between these monotone functions.
\begin{lem}\label{prop:indmonotone}
Let $\phi$ be a join-project plan  whose tree has $\ell$ leaves and let $t$ be a
tuple in $\tup(\phi)$.
\begin{enumerate}
\item The rank of 
  $I_{\phi(D),t}$ is at most $\ell$.
\item If $\phi_0$ is a subplan of $\phi$, then $I_{\phi(D),t}$ can be
  written as $I_{\phi(D),t}=\bigvee_{t'\in
    \tup(\phi_0)}(I_{\phi_0(D),t'} \wedge J_{t'})$, where each
  $J_{t'}$ is a monotone function of the random variables. Moreover,
  if $\phi'$ is obtained from $\phi$ by replacing $\phi_0$ with some
  other subplan $\phi'_0$ satisfying $A_{\phi_0}=A_{\phi'_0}$, then
  $I_{\phi'(D),t}=\bigvee_{t'\in \tup(\phi'_0)}(I_{\phi'_0(D),t'}
  \wedge J_{t'})$ with the same  functions $J_{t'}$.
\end{enumerate}
\end{lem}
\begin{proof}
  Statement 1 can be proved by a simple induction on the size of the tree
  of $\phi$. First, observe that the rank of disjunction of
  functions is at most the maximum of the ranks of the functions,
  while the rank of the conjunction of functions is at most the sum of
  the ranks of the functions. If the tree consists of a single leaf
  (i.e., $\phi$ consists of a single relation symbol), then
  $I_{\phi(D),t}$ is equal to one of the random variables, i.e., its
  rank is 1. If $\phi=\pi_X(\phi^*)$, then $I_{\phi(D),t}=
  \bigvee_{t'\in \tup(\phi^*), \pi_X(t')=\pi_X(t)} I_{\phi^*(D),t'}$. By
  induction, the rank of each $I_{\phi^*(D),t'}$ is at most $\ell$,
  hence the rank of this disjunction is also at most
  $\ell$. Finally, if $\phi=\phi_1 \bowtie \phi_2$, then
  $I_{\phi(D),t}=I_{\phi_1(D),\pi_{A_{\phi_1}}(t)} \wedge
  I_{\phi_2(D),\pi_{A_{\phi_2}}(t)}$. The number $\ell$ of leaves of
  $\phi$ is the sum of the number of leaves of $\phi_1$ and $\phi_2$,
  hence the rank of this conjunction is at most the number
  of leaves of $\phi$.

  To prove Statement 2, we build the function $I_{\phi(D),t}$ using
  disjunctions and conjunctions as in the previous paragraph, but the
  functions $I_{\phi_0(D),t'}$ are not decomposed any further. This
  way, $I_{\phi(D),t}$ is expressed as a monotone function of the
  random variables and of the functions $I_{\phi_0(D),t'}$ ($t'\in
  \tup(\phi_0)$). Thus $I_{\phi(D),t}$ can be written in the required
  form and it is clear that $J_{t'}$ does not depend on the structure
  of the subplan $\phi_0$.
\end{proof}

Let $f$ and $g$ be two functions on a set of independent 0-1
variables. Suppose that for each minterm of $f$, the probability that
$g=1$ on condition that the minterm is 1 is at least $p$. Somewhat
counterintuitively, this does not necessarily mean that the probability
of $g=1$ on condition that $f=1$ is at least $p$. For example,
consider two independent variables $x_1$ and $x_2$ with probability
$1/2$ of being 1, and let $g=x_1\wedge x_2$ and $f=x_1\vee x_2$. Now
$\Prob[ g=1\;|\;f=1 ]=1/3$, but $\Prob[ g=1\;|\;x_1=1
]=\Prob[ g=1\;|\;x_2=1 ]=1/2$, i.e., the probability is larger
conditioned on either minterm of $f$ than on $f$ itself. The following
lemma shows that if the functions have bounded rank, then we can bound
the ratio of these conditional probabilities.
\begin{lem}\label{lem:condprob}
Let $f,g$ be two monotone 0-1 functions of rank at most $\ell$ on a set
$V$ of independent 0-1
  random variables.
\begin{enumerate}
\item If the probability of 1 for each random variable
  is decreased by at most a factor $c>1$, then the expected value of
  $f$ is decreased by at most a factor $c^{\ell}$.
\item  If $\Prob[ g=1 \;|\; f^{(i)}=1 ]\ge p$ for every minterm
  $f^{(i)}$ of $f$, then $\Prob[ g=1 \;|\; f=1 ] \ge p\cdot 2^{-2\ell}$.
\end{enumerate}
\end{lem}
\begin{proof}
  To prove Statement 1, let $x_1:V\to \{0,1\}$ be a random assignment
  of the variables chosen according to the original probabilities, and
  let $x_2$ be a the independent random assignment, where $x_2(v)=1$
  with probability $1/c$ uniformly and independently for each $v$.
  Let $m$ be the number of minterms of $f$ and for $1\le i \le m$, let
  $f^{(i)}$ be the $i$-th minterm of $f$. Let $X_i$ be the event that
  $f^{(i)}(x_1)=1$ and $f^{(j)}(x_1)=0$ for every $j < i$, i.e., the
  $i$-th minterm is the first satisfied minterm of $f$. The events
  $X_i$ are disjoint, thus $\Prob[ f(x_1)=1 ]=\sum_{i=1}^m\Prob[ X_i
  ]$. Let $Y_i$ be the event that $f^{(i)}(x_2)=1$. Clearly, $X_i$ and
  $Y_i$ are independent and $\Prob[ Y_i ] \geq c^{-\ell}$.  Let
  $x_1\wedge x_2$ be the conjunction of $x_1$ and $x_2$; observe that
  in $x_1 \wedge x_2$, the probability of a variable being 1 is
  exactly $1/c$ times the original probability.  Each event $X_i\wedge
  Y_i$ implies $f(x_1\wedge x_2)=1$ (as it implies $f^{(i)}(x_1\wedge
  x_2)=1$) and these events are are disjoint (since the events $X_i$
  are already disjoint).  Therefore,
\[
\Prob[ f(x_1\wedge x_2)=1 ] \ge \sum_{i=1}^m\Prob[ X_i\wedge Y_i ]=
\sum_{i=1}^m\Prob[ X_i ] \Prob[ Y_i ] \geq \Prob[ f(x_1)=1 ] 
\cdot c^{-\ell},
\]
what we had to show.

For Statement 2, let $x_1$, $x_2$ be two independent random
assignments, chosen according to the original probabilities of the
random variables. Let $X_i$ be the event as above and let $Z_i$ be the
event that $g(x_2\vee m_i)=1$, where $m_i$ is the assignment that sets
the $i$-th minterm of $f$ to 1 and every other random variable to
0. Clearly, $X_i$ and $Z_i$ are independent. Observe that $\Prob[
Z_i ]=\Prob[ g(x_2)=1\;|\; f^{(i)}(x_2)=1 ]\ge p$, since $Z_i$
depends only on what $x_2$ assigns to the variables not in the $i$-th
minterm.

In order to bound the conditional probability $\Prob[ g(x_1)=1 \;|\;
f(x_1)=1 ]$, we need to bound the probability $\Prob[ g(x_1)=1 \wedge
f(x_1)=1 ]$. To bound this probability, we calculate $\Prob[ g(x_1
\vee x_2)=1 \wedge f(x_1 \vee x_2)=1 ]$ (where $x_1\vee x_2$ is the
assignment defined as the disjunction of $x_1$ and $x_2$). Observe
that for each variable $v$, $\Prob[ (x_1\vee x_2)(v)=1 ] \le 2\Prob[
x_1(v)=1 ]$.  Thus, by Statement 1, $\Prob[ g(x_1 \vee x_2)=1 \wedge
f(x_1\vee x_2)=1 ] \le 2^{2\ell} \Prob[ g(x_1)=1 \wedge f(x_1)=1 ]$
(here we used that the rank of the conjunction of two rank $\ell$
functions is at most $2\ell$). We can bound the conditional
probability now as follows:
\begin{align*}
\Prob[ g(x_1)=1 \;|\;
f(x_1)=1 ] &=\frac{\Prob[ f(x_1)=1 \wedge
g(x_1)=1 ]}{\Prob[ f(x_1)=1 ]}\\
&\ge
\frac{2^{-2\ell}\Prob[ f(x_1\vee x_2)=1 \wedge
g(x_1 \vee x_2)=1 ]}{\Prob[ f(x_1)=1 ]} & \text{(by Statement 1)}\\
&\ge
\frac{2^{-2\ell}\Prob[ f(x_1)=1 \wedge
g(x_1 \vee x_2)=1 ]}{\Prob[ f(x_1)=1 ]}&\text{(more restricted event)}\\
&=
\frac{2^{-2\ell}\sum_{i=1}^m\Prob[ X_i\wedge g(x_1\vee x_2)=1 ]}{\Prob[ f(x_1)=1 ]}
&\text{(by $\Prob(f(x_1)=1)=\sum_{i=1}^m \Prob[X_i]$)}\\
&\ge
\frac{2^{-2\ell}\sum_{i=1}^m\Prob[ X_i\wedge Z_i ]}{\Prob[ f(x_1)=1 ]}
&\text{($X_i\wedge Z_i$ implies $g(x_1\vee x_2)=1$)}\\
&=
\frac{2^{-2\ell}\sum_{i=1}^m\Prob[ X_i ] \Prob[ Z_i ]}{\Prob[ f(x_1)=1 ]}
&\text{($X_i$ and $Z_i$ are independent)}\\
&\ge
\frac{2^{-2\ell}\cdot \Prob[ f(x_1)=1 ] \cdot p}{\Prob[ f(x_1)=1 ]}
&\text{(by $\Prob(f(x_1)=1)=\sum_{i=1}^m \Prob[X_i]$)}\\
&=2^{-2\ell}p.
\end{align*}

\end{proof}

\subsubsection{Analysis of Step 1}
\label{sec:step1}
The analysis of Step 1 relies on the following lemma, which shows that
we do not get many additional tuples if we take the join of only those
relations whose attributes are in $C_S(A)$.

\begin{lem}\label{lem:extend}
  Let $S\subseteq \sigma$ be a set of relation names and $A\subseteq
  A_\sigma$ a set of attributes. Let $A^*=C_S(A)$ and let
  $S^*=S[A^*]$. For every $A^*$-tuple $t^*$,
\[
\Prob[ t^*\in \pi_{A^*}(\bowtie_{R\in S}R(D)) \;|\; t^*\in \bowtie
_{R\in S^*}R(D) ] \ge 1/2.
\] 
\end{lem}
\begin{proof}
  By definition, $t^*\in \pi_{A^*}(\bowtie_{R\in S}R(D))$ if and only
  if there is a $t\in \bowtie_{R\in S}R(D)$ with
  $\pi_{A^*}(t)=t^*$. Note that if $t^*\in \bowtie_{R\in S^*}R(D)$ and
  $\pi_{A^*}(t)=t^*$, then $t$ satisfies all the relations in $S^*$,
  hence the probability that such a $t$ is in $\bowtie_{R\in S}R(D)$
  (assuming $t^*\in \bowtie_{ R\in S^*}R(D)$) depends only on the
  relations in $S\setminus S^*$.  We claim that this conditional
  probability is equal to the probability that a certain query $Q'$
  with schema $\sigma'$ has at least one solution.  The query $Q'$ is
  over the attributes $A_\sigma \setminus A^*$. The schema $\sigma'$
  contains a relational symbol $R'$ for each $R\in S\setminus S^*$;
  the set of attributes of $R'$ is $A_R\setminus X^*$. We define the
  probability of placing a tuple into $R'$ as $p_{R'}(N)=p_{R}(N)$ for
  every $R'\in \sigma'$. It is not difficult to see that $\Prob[
  t^*\in \pi_{A^*}(\bowtie_{R\in S}R(D)) \;|\; t^*\in \bowtie_{R\in
    S^*}R(D) ]$ is equal to the probability that $Q'$ has at least one
  solution.

  Observe that if $A'$ is a subset of the attributes in $Q'$, then the
  relations in $\sigma'[A']$ were obtained from the relations in
  $S[A^*\cup A']\setminus S[A^*]$, which means that the weight of
  these relations is counted in $f_S(A^*\cup A')$ but not in
  $f_S(A^*)$. If the weight of the relations in $\sigma'[A']$ is
  greater than $|A'|(\log N-n-1)$, then $f_S(A^*\cup A')<f_S(A^*)$
  would follow, contradicting the minimality of
  $f_S(A^*)=f_S(C_S(A))$.  This means that the maximum density
  $\maxdensity$ of $Q'$ is at most $\log N-n-1$.  Writing $X :=
  |Q'(D)|$, by Proposition~\ref{prop:variance} the variance of $X$ is
\[
\Var[X]\le \Exp[X]^2\cdot (2^{n'}-1)2^{\maxdensity-\log
  N}<\Exp[X]^2\cdot (2^{n'}-1)2^{-(n+1)}\le \Exp[X]^2/2.
\]
Therefore, by Chebyshev's Inequality, the probability that there is no
solution can be bounded as
\[
\Prob[X=0]\le \Prob\left[ | X-\Exp[X]|\ge \Exp[X]\right] \le
\Var[X]/\Exp[X]^2\le 1/2.
\]
\end{proof}

Let $\phi'$ be the new join-project plan obtained from $\phi$ by an
application of Step 1.  For every subplan $\psi'$ of $\phi'$, there is
a corresponding subplan $\psi$ of $\phi$.  We claim that
\[
\Exp[ |\psi'(D)| ]\le 2^{2\ell+1} \Exp[ |\psi(D)| ],
\] 
where $\ell$ is the number of leaves in $\psi$.  If subplan $\psi$ is
disjoint from $\phi_0$, then $\psi'$ and $\psi$ are the same, and we
are done (e.g., if $\psi$ is the subplan rooted at $\pi_{\{d,f\}}$ in
Figure~\ref{fig:plan}). Thus we have to consider only two cases:
$\psi$ is either completely contained in $\phi_0$, or $\psi$ contains
$\phi_0$.

\medskip\noindent \textit{Case 1: } $\psi$ is contained in $\phi_0$
(e.g., $\psi$ is one of the join nodes below $\pi_{\{a,b\}}$ in
Figure~\ref{fig:plan}). Let $B$ be the set of all attributes of the
relations appearing in $\psi$. If $B\subseteq A^*$, then the
attributes of each relation appearing in $\psi$ are fully contained in
$A^*$ and we are done: $\psi'=\psi$. Otherwise, let $w$ (resp., $w'$)
be the total weight of the relations appearing in $\psi$ (resp.,
$\psi'$). Observe that $w-w'\le|B\setminus A^*|(\log N-n-1)$:
otherwise we would have $f_S(A^*\cup B)<f_S(A^*)$, contradicting the
minimality of $A^*$. Clearly, $A_{\psi'}\subseteq B\cap A^*$. Thus the
expected size of $|\psi'(D)|$ is
\[
\Exp[ |\psi'(D)| ]= 2^{|A_{\psi'}|\log N - w'}\le 2^{|B\cap A^*|\log N
  - w' + |B\setminus A^*|(\log N-n-1)-(w-w')} \le 2^{|B|\log N -
  w}=\Exp[ |\psi(D)| ].
\]

\medskip\noindent \textit{Case 2: } $\psi$ is not contained in
$\phi_0$, which implies that $\psi$ contains $\pi_A(\phi_0)$ as
subplan (including the possibility that $\psi'=\pi_A(\phi_0)$. As an
example, consider the projection $\pi_{\{a,c,f\}}$ in
Figure~\ref{fig:plan}. Note that $A_{\psi}=A_{\psi'}$ (because we are
above the projection $\pi_A$) and $\psi(D)\subseteq \psi'(D)$ follows
from $\pi_A(\phi_0(D))\subseteq \pi_A(\phi'_0(D))$. Let $\ell$ be the
number of leaves of $\psi$. We claim that for every tuple $t\in
\tup(\psi')$ we have $\Prob[ t\in \psi'(D) ] \le 2^{2\ell+1} \Prob[
t\in \psi(D) ]$, which implies $\Exp[ |\psi'(D)| ]\le 2^{2\ell+1}
\Exp[ |\psi(D)| ]$.  To prove this, we show that for every minterm
$I^{(i)}_{\psi'(D),t}$ of $I_{\psi'(D),t}$, we have 
\begin{equation}
\Prob[I_{\psi(D),t}=1\;|\;I^{(i)}_{\psi'(D),t}=1 ]\ge 1/2.
\end{equation} 
Thus by Lemma~\ref{lem:condprob}(2), we get $\Prob[
I_{\psi(D),t}=1\;|\;I_{\psi'(D),t}=1 ]\ge 1/2^{2\ell+1}$, what we
need.

By Lemma~\ref{prop:indmonotone}(2), $I_{\psi'(D),t}$ can be written as
\begin{equation}
  I_{\psi'(D),t}=\bigvee_{t'\in \pi_A(\tup(\phi'_0))} (I_{\pi_A(\phi'_0(D)),t'}\wedge
  J_{t'}).\label{eq:disj}
\end{equation}
Consider a particular minterm $I^{(i)}_{\psi'(D),t}$ for some $i$,
which is the product of a subset $V_i$ of the random variables.  If
$\prod V_i=1$, then $I_{\psi'(D),t}=1$, implying that
$I_{\pi_A(\phi'_0(D)),t'}\wedge J_{t'} =1$ for some tuple $t'\in
\tup(\pi_A(\phi'_0))$, which further implies that there is a tuple
$t''\in \tup(\phi'_0)$ such that $\pi_A(t'')=t'$ and
$I_{\phi'_0(D),t''}=1$. Let us fix such a $t'$ and $t''$. Since
$\phi'_0$ does not contain projections, $I_{\phi'_0(D),t''}$ is the
product of a set $V_i'$ of random variables. As
$I^{(i)}_{\psi'(D),t}=1$ implies $I_{\phi'_0(D),t''}=1$, we have
$V_i'\subseteq V_i$.  By a consequence of the FKG Inequality (see
Lemma~\ref{lem:poscor}),
\begin{equation}
  \Prob[ I_{\pi_A(\phi_0(D)),t'}=1 \;|\;\prod V_i=1 ] \ge 
  \Prob[ I_{\pi_A(\phi_0(D)),t'}=1\;|\;\prod V'_i=1 ]\label{eq:poscorr}
\end{equation}
(since $V'_i\subseteq V_i$ and $I_{\pi_A(\phi_0(D)),t'}$ is monotone).
Note that it is $\phi_0$ and not $\phi_0'$ that appears in
\eqref{eq:poscorr}.  By Lemma~\ref{prop:indmonotone}(2),
$I_{\psi(D),t}$ can be also written in the form \eqref{eq:disj}, hence
$I_{\pi_A(\phi_0(D)),t'}\wedge J_{t'}=1$ implies
$I_{\psi(D),t}=1$. Thus we have
\begin{align*}
\Prob[ I_{\psi(D),t}=1\;|\; \prod V_i=1 ]&
\ge
\Prob[ I_{\pi_A(\phi_0(D)),t'}=1\wedge J_{t'} =1\;|\; \prod V_i=1 ] \\
&=
\Prob[ I_{\pi_A(\phi_0(D)),t'}=1 \;|\; \prod V_i=1 ] \\
&\ge
\Prob[ I_{\pi_A(\phi_0(D)),t'}=1 \;|\; \prod V'_i=1 ] \\
&\ge
\Prob[ I_{\pi_{A^*}(\phi_0(D)),t''}=1 \;|\; \prod V'_i=1 ] \\
&=
\Prob[ t''\in \pi_{A^*}(\phi_0(D)) \;|\; t''\in \phi'_0(D) ] \\
&\ge 1/2.
\end{align*}
what we had to show. (The first inequality follows from the fact that
$I_{\pi_A(\phi_0(D)),t'}\wedge J_{t'}=1$ implies $I_{\psi(D),t}=1$;
the equality after that from the fact that $\prod V_i=1$ implies
$J_{t'}=1$; the second inequality follows from \eqref{eq:poscorr}; the
third inequality follows from $\pi_A(t'')=t'$; the
last inequality follows from Lemma~\ref{lem:extend}.)

\subsubsection{Analysis of Step 2}
\label{sec:step2}
The analysis of Step 2 relies on the following lemma, which shows that
extending the projection from $A$ to $A^*$ does not increase the
number of tuples too much: a tuple in the projection to $A$ does not
have too many extensions to $A^*$.

\begin{lem}\label{lem:projext}
  Let $S\subseteq \sigma$ be a set of relation names and $A \subseteq
  A_\sigma$ be a set of attributes. Let $A^*=C_S(A)$ and let
  $S^*=S[A^*]$.  For an $A$-tuple $t$, let $L_t$ be the set of those
  $A^*$-tuples $t^*\in \bowtie_{R\in S^*}R(D)$ that have
  $\pi_A(t^*)=t$.
\begin{enumerate} \itemsep=0pt
\item For every $A$-tuple $t$ we have $\Exp[ |L_t| \;|\;
  L_t\neq\emptyset ] \le 2^{n(n+2)}$.
\item For every $A$-tuple $t$ and $A^*$-tuple $t^*$ with
  $\pi_A(t^*)=t$ we have $\Prob[ t^*\in L_t \;|\; L_t\neq \emptyset ]
  \le 2^{n(n+2)}N^{-|A^*\setminus A|}$.
\end{enumerate}
\end{lem}
For the proof of Lemma~\ref{lem:projext}, we need the following probability bound:
\begin{lem}\label{lem:nonzeroprob}
For every real-valued random variable $X$,
\[
\Exp[ X \;|\; X\neq 0] \Exp[X] \le \Exp[X^2].
\]
\end{lem}
\begin{proof}
  Let $\lambda=\Prob[X\neq 0]$. If $\lambda = 0$ then $\Exp[X] =
  \Exp[X^2] = 0$ and we are done. Assume then $\lambda \not= 0$. Let
  $Y$ be the random variable defined by
  $\Prob[Y=a]=\Prob[X=a]/\lambda$ for every $a\neq 0$ and
  $\Prob[Y=0]=0$. Observe that $\Exp[Y]=\Exp[X]/\lambda$ and
  $\Exp[Y^2]=\Exp[X^2]/\lambda$.  By Jensen's Inequality (using that
  $f(x)=x^2$ is convex), we have $(\Exp[Y])^2\le
  \Exp[Y^2]$. Therefore,
\[
\Exp[ X \;|\; X\neq 0] \Exp[X] = \Exp[X]^2/\lambda =\lambda
(\Exp[Y])^2 \le \lambda \Exp[Y^2] = \Exp[X^2],
\]
what we had to show.
\end{proof}
Lemma~\ref{lem:nonzeroprob} shows that we can prove Lemma~\ref{lem:projext}(1) by bounding $\Exp[|L_t|^2]$.
\begin{proof}[Proof (of Lemma~\ref{lem:projext})]
  We show that
\[
\Exp\left[|L_t|^2
\right]\le \Exp[|L_t|]\cdot 2^{n(n+2)}
\] by an argument similar to the proof of Proposition~\ref{prop:variance}.
Statement (1) follows from Lemma~\ref{lem:nonzeroprob} with $X=|L_t|$.

Let $T$ be the set of all $A^*$-tuples whose projection to $A$ is $t$.
Let $X(t^*)$ be the indicator random variable corresponding to the
event $t^*\in L_t$. We need to bound
\[
\Exp\left[ |L_t|^2 \right] =\Exp\left[ \sum_{t^*_1,t^*_2\in
    T}X(t^*_1)X(t^*_2) \right]= \sum_{t^*_1,t^*_2\in
  T}\Exp\left[X(t^*_1)X(t^*_2)\right].
\]
For every $A\subseteq B\subseteq A^*$, let $F_B$ be the set of all
pairs $(t^*_1,t^*_2)\in T^2$ where $t^*_1$ and $t^*_2$ agree exactly
on $B$. We have that $|F_B|=N^{|B\setminus A|}\cdot N^{|A^*\setminus
  B|}\cdot (N-1)^{|A^*\setminus B|}\le N^{|B\setminus
  A|+2|A^*\setminus B|}=N^{|A^*\setminus A|+|A^*\setminus B|}$. Let
$w$ be the total weight of the relations in $S^*$ and let $w_B$ be the
total weight of the relations in $S^*[B]$. Observe that
$\Exp[|L_t|]=N^{|A^*\setminus A|}\cdot 2^{-w}$. Furthermore, we have
$w-w_B\ge |A^*\setminus B|(\log N-n-1)$: otherwise $f_S(B)$ would be
strictly less than $f_S(A^*)$, contradicting the minimality of
$A^*=C_S(A)$. For every $(t^*_1,t^*_2)\in F_B$, the event
$X(t^*_1)X(t^*_2)=1$ implies that for every relation $R\in S^*$, the
projections of $t^*_1$ and $t^*_2$ to the attributes of $R$ is in
$R(D)$. For different relations, these events are clearly independent.
For every $R\in S^*[B]$, the projections of $t^*_1$ and $t^*_2$ to the
attributes of $R$ are the same. Thus the probability that the
projections are in the relation for every $R\in S^*[B]$ is exactly
$2^{-w_B}$.  Consider now a relation in $R\in S^*\setminus
S^*[B]$. For such an $R$, the projections of $t^*_1$ and $t^*_2$ are
different, thus their appearance in $R(D)$ are independent events. As
the total weight of the relations in $S^*\setminus S^*[B]$ is $w-w_B$,
we get the following bound:
\[
\Exp\left[X(t^*_1)X(t^*_2) \right]\le 2^{-w_B-2(w-w_B)}.
\]
It follows that 
\begin{align*}
  \Exp\left[|L_t|^2 \right] &= \sum_{t^*_1,t^*_2\in
    T}\Exp\left[X(t^*_1)X(t^*_2)\right]= \sum_{B\subseteq
    A^*}\sum_{(t^*_1,t^*_2)\in F_B}\Exp\left[X(t^*_1)X(t^*_2)\right]
  \\
  & \le 2^{|A^*|}\cdot N^{|A^*\setminus A|+|A^*\setminus B|}\cdot 2^{-w_B-2(w-w_B)}=(N^{|A^*\setminus A|}\cdot 2^{-w})\cdot 2^{|A^*|}\cdot 2^{|A^*\setminus B|\log N-(w-w_B)}\\
  & \le \Exp[|L_t|]\cdot 2^{|A^*|}\cdot 2^{|A^*\setminus B|(n+1)}\le
  \Exp[|L_t|]\cdot 2^n\cdot 2^{n(n+1)} \le \Exp[|L_t|]\cdot
  2^{n(n+2)},
\end{align*}
what we had to show.

To prove the second statement, observe first that if we fix an
$A$-tuple $t$, then by symmetry, $\Prob[ t^*\in L_t \;|\;
L_t\neq\emptyset ]$ has the same value for every $A^*$-tuple $t^*$
with $\pi_A(t^*)=t$.  There are $N^{|A^*\setminus A|}$ such tuples
$t^*$ and the size of $L_t$ is the sum of the indicator variables
corresponding to these tuples. It follows that $\Prob[ t^*\in L_t
\;|\; L_t\neq\emptyset ] =\Exp[ |L_t| \;|\; L_t\neq\emptyset]
N^{-|A^*\setminus A|} \le 2^{n(n+2)}N^{-|A^*\setminus A|} $.
\end{proof}

Let $\phi''$ be the plan obtained from $\phi'$ after Step 2 and let
$\ell$ be the number of leaves of the plan.  We show that
\[\Exp[ |\psi''(D)| ] \le 2^{n(n+2)+2\ell} \Exp[ |\psi'(D)| ] \]
for every subplan $\psi''$ of $\phi''$ and corresponding subplan
$\psi'$ of $\phi'$.
We consider three cases.

{\em Case 1.}
If $\psi'$ is a subplan of $\psi'_0$, then
$\psi''=\psi'$.

{\em Case 2.} Suppose that $\psi'$ is a subplan strictly containing
$\pi_A(\phi'_0)$ such that the root of $\psi'$ is a projection and let
$\psi''$ be the corresponding subplan after applying Step 2 (e.g.,
consider the node $\pi_{\{a,c,f\}}$ in Figure~\ref{fig:plan}). Observe
that $A_{\psi'}=A_{\psi''}$: extra attributes can be introduced only
by the change from $\pi_A(\phi'_0)$ to $\pi_{A^*}(\phi'_0)$ and any
such extra attribute is either already present in $A_\psi'$ or
projected out by the projection in the root of $\psi'$. For example,
in Figure~\ref{fig:plan}, replacing $\pi_{\{a,b\}}$ with
$\pi_{\{a,b,c,d\}}$ has no effect on the
attributes of the nodes above the projection $\pi_{\{a,c,f\}}$ as
attribute $c$ already appears there and attribute $d$ is projected
out.

We show that $\psi''(D)\subseteq \psi'(D)$. Indeed, if $t\in
\psi''(D)$, then $\pi_{A^*}(t)\in \pi_{A^*}(\varphi'_0(D))$, which
implies $\pi_A(t)\in \pi_A(\varphi'_0(D))$. This implies $t\in
\psi'(D)$, since the other subplans of $\psi'$ did not change. Thus
Step 2 does not increase the size of $\psi'(D)$ if $\psi'$ is a
subplan containing $\pi_A(\phi'_0)$ and the root of $\psi'$ is a
projection. Furthermore, it also follows that if $\psi'$ is a subplan
containing $\pi_A(\phi'_0)$ such that there is a projection node above
$\pi_A(\phi'_0)$, then the size of $\psi'(D)$ does not increase.

{\em Case 3.}  The only remaining situation that we have to verify is
that if $\psi'$ is a subplan of $\phi'$ containing $\pi_A(\phi'_0)$
and having no projection node above $\pi_A(\phi'_0)$. For example,
consider any of the two join nodes between $\pi_{\{a,b\}}$ and
$\pi_{\{a,c,f\}}$ in Figure~\ref{fig:plan}.  We show that in this case
$\Exp[ |\psi''(D)| ] \le 2^{n(n+2)+2\ell} \Exp[ |\psi'(D)| ]$, where
$\psi''$ is the subplan of $\phi''$ corresponding to $\psi'$.

Note that $A^*\subseteq A_{\psi''}$, since $\psi''$ contains no
projections above $\pi_{A^*}$.  Let $C$ be $A_{\psi'}\setminus
(A^*\setminus A)=A_{\psi''}\setminus (A^*\setminus A)$.  If a tuple
$t$ is in $\psi''(D)$, then clearly $t_1:=\pi_{A^*}(t)$ is in
$\phi'_0(D)$ and $t_2:=\pi_{C}(t)$ is in $\pi_C(\psi'(D))$. Thus
\begin{align}
  \notag
  &\Prob[ t\in \psi''(D) ]\\
  &\le \Prob[ t_1 \in \phi'_0(D) \wedge t_2\in \pi_C(\psi'(D)) ]
  \notag\\
  &= \Prob[ \pi_A(t_1)\in \pi_A(\phi'_0(D)) ] \cdot \Prob[ t_1 \in
  \phi'_0(D) \;|\; \pi_A(t_1)\in \pi_A(\phi'_0(D)) ]\cdot \Prob[ t_2
  \in \pi_C(\psi'(D))\;|\;t_1 \in \phi'_0(D) ]. \label{eq:factors}
\end{align}
By Lemma~\ref{lem:projext}(2) with $t=t^*=t_1$, the second factor in
\eqref{eq:factors} is at most $2^{n(n+2)}\cdot N^{-|A^*\setminus A|}$.
To bound the third factor, observe that
\[
I_{\pi_A(\phi'_0(D)),\pi_A(t_1)}=\bigvee_{\substack{t'_1\in
    \tup(\phi'_0)\\\pi_A(t'_1)=\pi_A(t_1)}}I_{\phi'_0(D),t'_1}
\]  
by the definition of the projection. As the plan $\phi'_0$ does not
contain any projections, each term $I_{\phi'_0(D),t'_1}$ on the right
hand side is the product of random variables, i.e., they are the
minterms of $I_{\pi_A(\phi'_0(D)),\pi_A(t_1)}$.  The conditional
probability
\[
\Prob[ I_{\pi_C(\psi'(D)),t_2}=1\;|\; I_{\phi'_0(D),t'_1}=1 ]
\]
is the same for every $t'_1\in \tup(\phi'_0)$ with
$\pi_A(t'_1)=\pi_A(t_1)=\pi_A(t_2)$, since the common attributes of
$t'_1$ and $t_2$ are exactly in $A$.  Therefore, we can use
Lemma~\ref{lem:condprob}(2) to show that the probability conditioned
on a minterm of $I_{\pi_A(\phi'_0(D)),\pi_A(t_1)}$ can be bounded by
the probability conditioned on $I_{\pi_A(\phi'_0(D)),\pi_A(t_1)}$
itself. In particular, we bound the probability of the event
$I_{\pi_C(\psi'(D)),t_2}=1$ conditioned on the minterm corresponding
to $t_1$, which is exactly the third factor appearing in
\eqref{eq:factors}:
\[
\Prob[ I_{\pi_C(\psi'(D)),t_2}=1\;|\; I_{\phi'_0(D),t_1}=1 ] \le
2^{2\ell}\Prob[ I_{\pi_C(\psi'(D)),t_2}=1\;|\;
I_{\pi_A(\phi'_0(D)),\pi_A(t_1)}=1 ].
\]

Therefore, continuing \eqref{eq:factors}, we have
\begin{align*}
  \Prob[ t\in \psi''(D) ] &\le \Prob[ \pi_A(t_1)\in \pi_A(\phi'_0(D))
  ] \cdot 2^{n(n+2)}N^{-|A^*\setminus A|} \cdot
  2^{2\ell}\Prob[ t_2\in \pi_C(\psi'(D)) \;|\; \pi_A(t_1)\in \pi_A(\phi'_0(D)) ]\\
  &\le
  2^{n(n+2)+2\ell}N^{-|A^*\setminus A|}\cdot \Prob[ t_2\in \pi_C(\psi'(D)) ]\\
  &= 2^{n(n+2)+2\ell}N^{-|A^*\setminus A|}\cdot \Exp[
  |\pi_C(\psi'(D))| ] N^{-|C|}.
\end{align*}
The last equality follows from the fact that the probability $\Prob[ t_2
\in \pi_C(\pi'(D)) ]$ is the same for every fixed $t_2$ by symmetry.
Thus the expected size of $\psi''(D)$ is
\begin{align*}
\Exp[ |\psi''(D)| ] &=N^{|A_{\psi''}|}\cdot 
2^{n(n+2)+2\ell}N^{-|A^*\setminus A|}\cdot
\Exp[ |\pi_C(\psi'(D))| ] N^{-|C|}\\
&=
2^{n(n+2)+2\ell}  \Exp[ |\pi_C(\psi'(D))| ]\\ 
&\le 
2^{n(n+2)+2\ell}  \Exp[ |\psi'(D)| ].
\end{align*}

\section{Conclusions}
We have conducted a theoretical study of database queries from the
viewpoint of bounding or estimating the size of the answer. In the
worst case model, we showed that the fractional edge cover number, or
more generally, the solutions of certain linear programs can be used
to obtain fairly tight bounds. In the random database model, we
investigated bounds on the expected size and whether the number of
solutions is well concentrated around the expectation. Perhaps the
most interesting message of the paper is that from the viewpoint of
worst-case size, join-project plans can be significantly more efficent
than join plans for the same query, while in the average case model
every join-project plan can be turned into a join plan with only a
bounded loss of performance.

Let us mention two possible directions in which our results could be
further developped. First, one can introduce functional dependencies
into the model and generalize the bounds to take these restrictions on
the relations into account. This has been investigated in
\cite{GottlobLeeValiant2009} and \cite{ValiantValiant2010}, but the
problem has not been fully resolved yet. Another direction to
investigate is to understand concentration bounds for join-project
plans. In particular, Theorem~\ref{thm:removeproject} states a bound
only on the expected value, but it does not give an upper bound on
size that holds with high probability. It remains a challenging
problem to prove a variant of Theorem~\ref{thm:removeproject} saying
that if every subplan of the join-project plan has bounded size with
high probability, then there is also a join plan with this property.

\bibliographystyle{plain}
\bibliography{queries-revised}

\begin{thebibliography}{10}

\bibitem{abihulvia95}
S.~Abiteboul, R.~Hull, and V.~Vianu.
\newblock {\em Foundations of Databases}.
\newblock Addison-Wesley, 1995.

\bibitem{AlonSpencerBook}
N.~Alon and J.~Spencer.
\newblock {\em The Probabilistic Method}.
\newblock John Wiley, second edition, 1992.

\bibitem{chau98}
S.~Chaudhuri.
\newblock An overview of query optimization in relational systems.
\newblock In {\em Proceedings of the seventeenth ACM SIGACT-SIGMOD-SIGART
  Symposium on Principles of Database Systems}, pages 34--43, 1998.

\bibitem{MR859293}
F.~R.~K. Chung, R.~L. Graham, P.~Frankl, and J.~B. Shearer.
\newblock Some intersection theorems for ordered sets and graphs.
\newblock {\em J. Combin. Theory Ser. A}, 43(1):23--37, 1986.

\bibitem{flufrigro02}
J.~Flum, M.~Frick, and M.~Grohe.
\newblock Query evaluation via tree-decompositions.
\newblock {\em Journal of the ACM}, 49(6):716--752, 2002.

\bibitem{frikah98}
E.~Friedgut and J.~Kahn.
\newblock On the number of copies of a hypergraph in another.
\newblock {\em Israel Journal of Mathematics}, 105:251--256, 1998.

\bibitem{garullwid99}
H.~Garcia-Molina, J.~Widom, and J.D. Ullman.
\newblock {\em Database System Implementation}.
\newblock Prentice-Hall, 1999.

\bibitem{GottlobLeeValiant2009}
Georg Gottlob, Stephanie~Tien Lee, and Gregory Valiant.
\newblock Size and treewidth bounds for conjunctive queries.
\newblock In {\em Proceedings of Symposium on Principles on Database Systems
  (PODS)}, 2009.

\bibitem{gra93}
G.~Graefe.
\newblock Query evaluation techniques for large databases.
\newblock {\em ACM Computing Surveys}, 25, 1993.

\bibitem{gromar06}
M.~Grohe and D.~Marx.
\newblock Constraint solving via fractional edge covers.
\newblock In {\em Proceedings of the of the 17th Annual ACM-SIAM Symposium on
  Discrete Algorithms}, pages 289--298, 2006.

\bibitem{MR1687331}
Johan H{\aa}stad.
\newblock Clique is hard to approximate within {$n\sp {1-\epsilon}$}.
\newblock {\em Acta Math.}, 182(1):105--142, 1999.

\bibitem{KimVu}
J.~H. Kim and V.~H. Vu.
\newblock Concentration of multivariate polynomials and its applications.
\newblock {\em Combinatorica}, 20(3):417--434, 2000.

\bibitem{PapadimitriouBook}
C.~H. Papadimitriou.
\newblock {\em Computational Complexity}.
\newblock Addison Wesley, 1994.

\bibitem{rha01}
J.~Rhadakrishnan.
\newblock Entropy and counting.
\newblock At http://www.tcs.tifr.res.in/$\sim$jaikumar/mypage.html.

\bibitem{ValiantValiant2010}
Gregory Valiant and Paul Valiant.
\newblock Size bounds for conjunctive queries with general functional
  dependencies.
\newblock In {\em http://arxiv.org/abs/0909.2030v2}, 2010.

\bibitem{vaz01}
V.V. Vazirani.
\newblock {\em Approximation Algorithms}.
\newblock Springer-Verlag, 2001.

\end{thebibliography}

\newpage
\appendix

\end{document}